\documentclass{article}

\usepackage[utf8]{inputenc}
\usepackage[hidelinks]{hyperref}
\usepackage{amsmath}
\usepackage{amsthm}
\usepackage{amssymb}
\usepackage{amsfonts}
\usepackage{mathpartir}
\usepackage{scalerel}
\usepackage{stmaryrd}
\usepackage{authblk}

\theoremstyle{definition}
\newtheorem{definition}{Definition}
\newtheorem{example}{Example}

\theoremstyle{theorem}
\newtheorem{theorem}{Theorem}
\newtheorem{lemma}{Lemma}

\newcommand{\refl}{\mathsf{refl}}
\newcommand{\id}{\mathsf{id}}
\newcommand{\Con}{\mathsf{Con}}
\newcommand{\Sub}{\mathsf{Sub}}
\newcommand{\Tm}{\mathsf{Tm}}
\newcommand{\Ty}{\mathsf{Ty}}
\newcommand{\U}{\mathsf{U}}
\newcommand{\El}{\mathsf{El}}
\newcommand{\Id}{\mathsf{Id}}
\newcommand{\proj}{\mathsf{proj}}
\renewcommand{\tt}{\mathsf{tt}}
\newcommand{\blank}{\mathord{\hspace{1pt}\text{--}\hspace{1pt}}}
\newcommand{\ra}{\rightarrow}
\newcommand{\Set}{\mathsf{Set}}

\newcommand{\ToS}{\mathsf{ToS}}
\newcommand{\ext}{\triangleright}
\newcommand{\emptycon}{\scaleobj{.75}\bullet}
\newcommand{\Pii}{\Pi}
\newcommand{\appi}{\mathsf{app}}
\newcommand{\lami}{\mathsf{lam}}
\newcommand{\Pie}{\Pi^{\mathsf{ext}}}
\newcommand{\appe}{\mathsf{app^{ext}}}
\newcommand{\lame}{\mathsf{lam^{ext}}}
\newcommand{\Piinf}{\Pi^{\mathsf{inf}}}
\newcommand{\appinf}{\mathsf{app^{inf}}}
\newcommand{\laminf}{\mathsf{lam^{inf}}}
\newcommand{\appitt}{\mathop{{\scriptstyle @}}}
\newcommand{\Refl}{\mathsf{Refl}}

\newcommand{\Sig}{\mathsf{Sig}}
\newcommand{\ToSSig}{\mathsf{ToSSig}}
\newcommand{\Inductive}{\mathsf{Inductive}}
\newcommand{\Subtype}{\mathsf{Subtype}}
\newcommand{\subtype}{\mathsf{subtype}}
\newcommand{\NatSig}{\mathsf{NatSig}}
\newcommand{\mi}[1]{\mathit{#1}}
\newcommand{\Sg}{\Sigma}
\newcommand{\flCwF}{\mathsf{flCwF}}
\newcommand{\Kfam}{\mathsf{K}}
\newcommand{\p}{\mathsf{p}}
\newcommand{\q}{\mathsf{q}}
\newcommand{\K}{\mathsf{K}}
\newcommand{\lamK}{\mathsf{lam}^{\K}}
\newcommand{\appK}{\mathsf{app}^{\K}}
\newcommand{\A}{\mathsf{A}}
\newcommand{\D}{\mathsf{D}}
\renewcommand{\S}{\mathsf{S}}
\newcommand{\arri}{\Rightarrow}
\newcommand{\arre}{\Rightarrow^{\mathsf{ext}}}
\newcommand{\arrinf}{\Rightarrow^{\mathsf{inf}}}
\newcommand{\syn}{\mathsf{syn}}
\newcommand{\SynSig}{\mathsf{SynSig}}
\newcommand{\bCon}{\boldsymbol{\Con}}
\newcommand{\bTy}{\boldsymbol{\Ty}}
\newcommand{\bSub}{\boldsymbol{\Sub}}
\newcommand{\bTm}{\boldsymbol{\Tm}}
\newcommand{\bGamma}{\boldsymbol{\Gamma}}
\newcommand{\bDelta}{\boldsymbol{\Delta}}
\newcommand{\bsigma}{\boldsymbol{\sigma}}
\newcommand{\bdelta}{\boldsymbol{\delta}}
\newcommand{\bepsilon}{\boldsymbol{\epsilon}}
\newcommand{\bt}{\boldsymbol{t}}
\newcommand{\bu}{\boldsymbol{u}}
\newcommand{\bA}{\boldsymbol{A}}
\newcommand{\ba}{\boldsymbol{a}}
\newcommand{\bb}{\boldsymbol{b}}
\newcommand{\bB}{\boldsymbol{B}}
\newcommand{\bid}{\boldsymbol{\id}}
\newcommand{\bemptycon}{\boldsymbol{\emptycon}}
\newcommand{\bSet}{\boldsymbol{\Set}}
\newcommand{\bU}{\boldsymbol{\U}}
\newcommand{\bEl}{\boldsymbol{\El}}
\newcommand{\bPii}{\boldsymbol{\Pi}}
\newcommand{\bPie}{\boldsymbol{\Pie}}
\newcommand{\bPiinf}{\boldsymbol{\Piinf}}

\newcommand{\bId}{\boldsymbol{\Id}}
\newcommand{\bM}{\boldsymbol{\mathsf{M}}}
\newcommand{\bT}{\boldsymbol{\mathsf{T}}}
\newcommand{\bS}{\boldsymbol{\mathsf{S}}}

\newcommand{\bD}{\boldsymbol{\mathsf{D}}}
\newcommand{\bI}{\boldsymbol{\mathsf{I}}}
\newcommand{\ul}[1]{\underline{#1}}
\newcommand{\ulGamma}{\ul{\Gamma}}
\newcommand{\ulDelta}{\ul{\Delta}}
\newcommand{\ulgamma}{\ul{\gamma}}
\newcommand{\ulOmega}{\ul{\Omega}}
\newcommand{\uldelta}{\ul{\delta}}
\newcommand{\ulsigma}{\ul{\sigma}}
\newcommand{\ulnu}{\ul{\nu}}

\newcommand{\ult}{\ul{t}}
\newcommand{\ulu}{\ul{u}}
\newcommand{\ulA}{\ul{A}}
\newcommand{\ula}{\ul{a}}
\newcommand{\ulB}{\ul{B}}

\newcommand{\coe}{\mathsf{coe}}
\newcommand{\coh}{\mathsf{coh}}
\newcommand{\llb}{\llbracket}
\newcommand{\rrb}{\rrbracket}

\title{Large and Infinitary Quotient Inductive-Inductive Types}

\author{Andr{\'a}s Kov{\'a}cs}
\author{Ambrus Kaposi}
\affil{E\"otv\"os Lor\'and University}

\begin{document}

\maketitle

\begin{abstract}
Quotient inductive-inductive types (QIITs) are generalized inductive types which
allow sorts to be indexed over previously declared sorts, and allow usage of
equality constructors. QIITs are especially useful for algebraic descriptions of
type theories and constructive definitions of real, ordinal and surreal
numbers. We develop new metatheory for large QIITs, large elimination, recursive
equations and infinitary constructors. As in prior work, we describe QIITs using
a type theory where each context represents a QIIT signature. However, in our
case the theory of signatures can also describe its own signature, modulo
universe sizes. We bootstrap the model theory of signatures using
self-description and a Church-coded notion of signature, without using
complicated raw syntax or assuming an existing internal QIIT of signatures. We
give semantics to described QIITs by modeling each signature as a finitely
complete CwF (category with families) of algebras. Compared to the case of
finitary QIITs, we additionally need to show invariance under algebra
isomorphisms in the semantics. We do this by modeling signature types as
isofibrations. Finally, we show by a term model construction that every QIIT is
constructible from the syntax of the theory of signatures.
\end{abstract}

\section{Introduction}

The aim of this work is to provide theoretical underpinning to a general notion
of inductive types, called quotient inductive-inductive types (QIITs). QIITs are
of interest because there are many commonly used mathematical structures, which
can be conveniently described as QIITs in type theory, but cannot be defined as
less general inductive types, or doing so incurs large encoding overhead.

Categories are a prime example for a structure which is described by a
quotient inductive-inductive signature. Signatures for QIITs allow
having multiple sorts, with later ones indexed over previous ones, and
equations as well. We need both features in order to write down the
following signature of categories.
\begingroup
\allowdisplaybreaks
\begin{alignat*}{3}
  & \mi{Ob}  && : \Set \\
  & \mi{Mor} && : \mi{Ob}\ra \mi{Ob}\ra\Set \\
  & \blank\circ\blank && : \mi{Mor}\,J\,K\ra \mi{Mor}\,I\,J\ra \mi{Mor}\,I\,K \\
  & \mi{id} && : \mi{Mor}\,I\,I \\
  & \mi{ass} && : (f\circ g)\circ h = f\circ(g\circ h) \\
  & \mi{idl} && : \mi{id}\circ f = f \\
  & \mi{idr} && : f\circ \mi{id} = f
\end{alignat*}
\endgroup
The benefit of having a QII signature is getting a model theory ``for free'',
from the metatheory of QIITs. This model theory includes a category of algebras
which has an initial object and also some additional structure. For the
signature of categories, we get the empty category as the initial object, but it
is common to consider categories with more structure, which have more
interesting initial models.

Algebraic notions of models of type theories are examples for this. Here,
initial models represent syntax, and initiality corresponds to induction on
syntax. Several variants have been used, from contextual categories \cite{gat}
and comprehension categories \cite{jacobs1993comprehension} to categories with
families \cite{internaltt} (CwF).

A prime motivation of the current work is to further develop QIITs as a
framework for the metatheory of type theories, to cover more theories and
support more applications. To this end, we extend the syntax and semantics of
QIITs as previously described in the literature \cite{kaposi2019constructing,
  qiits, dijkstra2017quotient}, with the following features.

\begin{enumerate}
  \item
  \textbf{Large constructors, large elimination} and algebras at different
  universe levels. This fills in an important formal gap; large models are
  routinely used in the metatheory of type theories, but they have not been
  presented explicitly in previous QIIT literature. For example, interpreting
  syntactic contexts as sets already requires a notion of large models.
  \item
  \textbf{Infinitary constructors}.
  This allows specification of infinitely branching trees. Examples of
  infinitary QIITs in previous works include real, surreal numbers
  \cite{HoTTbook}, ordinal numbers \cite{lumsdaineShulman} and a partiality
  monad \cite{partiality}. Of special note here is that the theory of QIIT
  signatures is itself large and infinitary, thus it can ``eat itself'',
  i.e.\ include its own signature and provide its own metatheory. This was not
  possible previously in \cite{kaposi2019constructing}, where only finitary
  QIITs were described. In this paper we use self-representation to bootstrap
  the model theory of signatures, without having to assume any pre-existing
  internal syntax.
  \item
  \textbf{Recursive equations}, i.e. equations appearing as assumptions of
  constructors. These have occurred previously in syntaxes of cubical type
  theories, as boundary conditions \cite{cchm, angiuli2016computational,
    angiuli2018cartesian}.
\end{enumerate}

To provide semantics, we show that for each signature, there is a CwF (category
with families) of algebras, extended with $\Sigma$-types, extensional identity,
and constant families. This additional structure corresponds to a type-theoretic
flavor of finite limits, as it was shown in \cite{clairambault2014biequivalence}
that the category of such CwFs is biequivalent to the category of finitely
complete categories.

Compared to the case of finitary QIITs, the addition of infinitary constructors
and recursive equations requires a significant change in semantics: instead of
strict CwF morphisms, we need to consider weak ones, and instead of modeling
types as displayed CwFs, we need to model them as CwF isofibrations. The latter
amounts to showing that signature extension respects algebra isomorphisms.

We also show, by a term model construction, that all QIITs are reducible to the
syntax of signatures. This construction also essentially relies on invariance under
isomorphisms.

\subsection{Outline of the Paper}

In Section \ref{sec:metatheory}, we describe the metatheory used in the rest of
the paper. In Section \ref{sec:tos}, we introduce the theory of QIIT
signatures. In Section \ref{sec:semantics} we give categorical semantics to
signatures. In Section \ref{sec:tossig} we build model theory for the theory
of QIIT signatures. In Section \ref{sec:termmodels} we give a term model construction
of QIITs. We discuss related work and conclude in Sections \ref{sec:relatedwork}-\ref{sec:conclusion}.

\section{Metatheory}
\label{sec:metatheory}

The metatheory used in this paper is extensional type theory, extended with a
form of cumulativity and an external notion of universe polymorphism. We refer
to this theory as cETT. We review the used features and notations
in the following.

\subsection{Core Extensional Theory}

We have Russell-style predicative universes $\Set_i$ indexed by natural numbers,
dependent functions as $(x : A)\ra B$, and dependent pairs as $(x : A)\times
B$. We sometimes leave parameters implicit in dependent function types,
e.g.\ write $\id : A \ra A$ instead of $\id : (A : \Set_i)\ra A \ra A$. We also
use subscripts as a field projection notation for iterated pairs. For example,
for $t : (A : \Set_i) \times (B : \Set_i) \times (f : A \ra B)$, we use $B_t$ to
denote the projection of the second component. Sometimes we omit the subscript
if it is clear from context. When we write ``exists'' in this paper, we always
mean chosen structure given by a $\Sigma$-type.

Both for function types and $\Sigma$, the output universe level is given as the
maximum of the levels of the constituent types, e.g.\ $(x : A)\ra B : \Set_{\max(i,j)}$
when $A : \Set_i$ and $B : \Set_j$.

We write propositional equality as $t = u$, with $\refl_t$ for reflexivity. We
have equality reflection and uniqueness of identity proofs (UIP). The unit type
is $\top : \Set_0$, with inhabitant $\tt$.

\subsection{Cumulativity}
\label{sec:cumulativity}

We use cumulative universes and cumulative subtyping as described
in \cite{timany2018cumulative}. Concretely, we have a $\blank\leq\blank$ subtyping
relation on types, specified by the following rules:
\begin{mathpar}
  \inferrule*{i \leq j}
             {\Gamma \vdash \Set_i \leq \Set_j}

  \inferrule*{\Gamma,\,x : A \vdash B \leq B'}
             {\Gamma \vdash (x : A)\ra B \leq (x : A)\ra B'}

  \inferrule*{\Gamma\vdash A \leq A' \\ \Gamma,\,x : A \vdash B \leq B'}
             {\Gamma \vdash (x : A)\times B \leq (x : A') \times B'}

  \inferrule*{\\}
             {\Gamma \vdash A \leq A}

  \inferrule*{\Gamma \vdash A \leq B \\ \Gamma\vdash B \leq C}
             {\Gamma \vdash A \leq C}

  \inferrule*{\Gamma\vdash A \leq A' \\ \Gamma\vdash t : A}
             {\Gamma \vdash t : A'}
\end{mathpar}

Additionally, we have an internal $\Subtype$ type, which internalizes subtyping,
analogously to how $\blank=\blank$ internalizes definitional equality. Hence, we
have analogous reflection and uniqueness rules.
\begin{mathpar}
  \inferrule*{\Gamma \vdash A : \Set_i \\ \Gamma \vdash B : \Set_j}
             {\Gamma \vdash \Subtype\,A\,B : \Set_{\max(i, j)}}

  \inferrule*{\Gamma \vdash A \leq B}
             {\Gamma \vdash \subtype : \Subtype\,A\,B}

  \inferrule*{\Gamma \vdash t : \Subtype\,A\,B}
             {\Gamma \vdash A \leq B}

  \inferrule*{\Gamma \vdash t : \Subtype\,A\,B \\ \Gamma \vdash u : \Subtype\,A\,B}
             {\Gamma \vdash t \equiv u}
\end{mathpar}

We use cumulativity to reduce bureaucratic overhead when dealing with
constructions at different universe levels. The internal $\Subtype$ is used in
Section \ref{sec:termmodels} to prove cumulativity for general QIIT
algebras. For example, consider natural number algebras at level $i$, given as
the $\Sigma$-type $\mi{NatAlg}_i :=(\mi{Nat} : \Set_i)\times
\mi{Nat}\times(\mi{Nat} \ra \mi{Nat})$. It follows from the subtyping rules that
$i \leq j$ implies $NatAlg_{i} \leq \mi{NatAlg}_{j}$. However, cumulativity for
arbitrary QIIT algebras does not follow judgmentally; it can only be proven by
induction on signatures, hence the need for $\Subtype$.

Internal subtyping is not included in \cite{timany2018cumulative}, but it can be
justified by the set-theoretic model given there.

\subsection{Universe Polymorphism}
\label{sec:universe_polymorphism}

We need to talk about constructions at arbitrary universe levels. For the sake
of simplicity, we do not assume a notion of universe polymorphism in cETT,
instead we quantify over levels in an unspecified theory outside of cETT. Hence,
a universe polymorphic cETT term is understood as a $\mathbb{N}$-indexed family
of cETT terms. We reuse the notation of cETT functions for
universe polymorphism, e.g.\ as in the following function:
\[
\lambda\,i.\,\Set_i : (i : \mathbb{N})\ra \Set_{i+1}
\]

\section{QIIT Signatures}
\label{sec:tos}

Signatures are given as contexts in a certain type theory, the theory of
signatures. We shall abbreviate it as ToS. However, ToS turns out to be a large
infinitary QIIT itself, and we would like to define ToS and a notion of
signature without referring to QIITs, only using features present in cETT.

In previous works by Cartmell \cite{gat} and Sterling
\cite{sterling2019algebraic}, signatures for generalized algebraic theories are
defined using raw syntax together with well-formedness relations. In this way,
signatures can be specified without already assuming the existence of GATs or
QIITs. However, raw syntax is notoriously difficult to work with, and we prefer
to avoid it altogether.

At this point, we do not actually need \emph{syntactic} signatures, which make
it possible to do induction on signatures. We only need a way to write down
well-formed signatures, and interpret them in arbitrary models of ToS. For this,
a weak Church-like encoding suffices, where a signature is given as a typing
context in an arbitrary model of ToS. For this, we first need to specify the
notion of ToS models. However, this is the \emph{only} piece of information
about ToS which we need to manually provide. Other concepts such as
homomorphisms of ToS models and ToS-induction, will be derived from the
semantics of signatures and self-description in Section \ref{sec:tossig}.

\begin{definition}[Notion of model for the theory of signatures]\label{def:tos}
For levels $i$ and $j$, $\ToS_{i,j} : \Set_{\max(i+1,\,j+1)}$ is a cETT type whose
elements are ToS models (or ToS-algebras). $\ToS_{i,j}$ is an iterated
$\Sigma$-type, containing all of the following components.

A \textbf{category with families} (CwF), where all four underlying sets (of
objects, morphisms, types and terms) are in $\Set_i$. Following notation in
\cite{kaposi2019constructing}, we denote these respectively as $\Con : \Set_i$,
$\Sub : \Con \ra \Con \ra \Set_i$, $\Ty : \Con \ra \Set_i$ and $\Tm : (\Gamma :
\Con) \ra \Ty\,\Gamma \ra \Set_i$. We use $\id$ and $\blank\circ\blank$ to
denote identity and composition for substitution. We denote the empty context as
$\emptycon : \Con$, and the unique substitution into the empty context as
$\epsilon : \Sub\,\Gamma\,\emptycon$. Context extension is $\blank\ext\blank :
(\Gamma : \Con)\ra \Ty\,\Gamma \ra \Con$. Substitution on types and terms is
written as $\blank[\blank]$. Projections are noted as $\p :
\Sub\,(\Gamma\,\ext\,A)\,\Gamma$ and $\q : \Tm\,(\Gamma\,\ext\,A)\,(A[\p])$, and
substitution extension is $\blank,\blank : (\sigma : \Sub\,\Gamma\,\Delta)\ra
\Tm\,\Gamma\,(A[\sigma])\ra \Sub\,\Gamma\,(\Delta\ext A)$.

A \textbf{universe} $\U : \Ty\,\Gamma$ with decoding $\El : (a :
\Tm\,\Gamma\,\U) \ra \Ty\,\Gamma$.

\textbf{Inductive function space} $\Pii : (a : \Tm\,\Gamma\,\U) \ra
\Ty\,(\Gamma\,\ext\,\El\,a) \ra \Ty\,\Gamma$, with application as $\appi :
\Tm\,\Gamma\,(\Pii\,a\,B)\ra \Tm\,(\Gamma\,\ext\,\El\,a)\,B$ and its
inverse $\lami$.

\textbf{External function space} $\Pie : (A : \Set_j)\ra(A \ra \Ty\,\Gamma)\ra
\Ty\,\Gamma$, with $\appe : \Tm\,\Gamma\,(\Pie\,A\,B)\ra((x : A)\ra
\Tm\,\Gamma\,(B\,x))$ and its inverse $\lame$.

\textbf{Infinitary function space} $\Piinf : (A : \Set_j)\ra(A \ra
\Tm\,\Gamma\,\U)\ra \Tm\,\Gamma\,\U$, with $\appinf :
\Tm\,\Gamma\,(\El\,(\Piinf\,A\,b))\ra((x : A)\ra \Tm\,\Gamma\,(\El\,(b\,x)))$ and
its inverse $\laminf$.

An \textbf{identity type} $\Id : (a : \Tm\,\Gamma\,\U)\ra
\Tm\,\Gamma\,(\El\,a)\ra\Tm\,\Gamma\,(\El\,a)\ra \Tm\,\Gamma\,\U$, with $\Refl :
(t : \Tm\,\Gamma\,(\El\,a))\ra \Tm\,\Gamma\,(\El\,(\Id\,a\,t\,t))$, equality reflection and UIP.

\end{definition}
In the above listing, we omit equations for substitution and
$\beta\eta$-conversion, but these should be understood to be also part of
$\ToS_{i,j}$.

\emph{Notational conventions.} We name elements of $\Con$ as $\Gamma$, $\Delta$,
$\Theta$, elements of $\Sub\,\Gamma\,\Delta$ as $\sigma$, $\delta$, $\nu$,
elements of $\Ty\,\Gamma$ as $A$, $B$, $C$, and elements of $\Tm\,\Gamma\,A$ as
$t$, $u$, $v$. CwF components by default support de Bruijn indices, which are
not easily readable. We use instead a nameful notation for binders in context
extension, $\Pii$ and $\lami$, e.g.\ as $(\emptycon\,\ext (a : \U) \ext (t :
\El\,a))$. We also define a type-theoretic flavor of $\appi$ for convenience:
\begin{alignat*}{3}
  & \blank\appitt\blank && :
      \Tm\,\Gamma\,(\Pii\,a\,B)\ra
      (u : \Tm\,\Gamma\,(\El\,a)) \ra \Tm\,\Gamma\,(B[\id, u])\\
  & t\appitt u && := (\appi\,t)[\id, u]
\end{alignat*}
We abbreviate non-dependent inductive $\Pii$ as $\blank\arri\blank$, and
likewise we use $\blank\arre\blank$ and $\blank\arrinf\blank$ for non-dependent
external and infinitary functions.


\begin{definition}[Notion of signature]\label{def:signature}
A QIIT signature at level $j$ is a context in an arbitrary $M : \ToS_{i,j}$
model. We define the type of such signatures as follows:
\[
  \Sig_j := (i : \mathbb{N})\ra(M : \ToS_{i,j})\ra \Con_M
\]

Here, $j$ refers to the level of external types appearing in the signature, in
the domains of $\Pie$ and $\Piinf$ functions, while the quantified $i$ level is
required to allow interpreting a signature in arbitrary-sized ToS models. Note
that $\Sig_j$ is universe-polymorphic, so it is a family of cETT types and it is
not in any cETT universe.

\begin{example}{
    Signature for natural numbers. Here, no external types appear, so the level
    can be chosen as $0$.}
\begin{alignat*}{3}
  & \NatSig : \Sig_0 \\
  & \NatSig := \lambda (i :\mathbb{N})(M : \ToS_{i,0}).\\
  & \hspace{1em}(\emptycon_M\,\ext_M\, (N : \U_M) \,\ext_M\,(\mi{zero} : \El_M\,N)\\
  & \hspace{2.67em}\ext_M (\mi{suc} : N\arri_M\El_M\,N))
\end{alignat*}
\end{example}

With this, we are able to specify QIITs, and we can also interpret each
signature in an arbitrary ToS model, by applying a signature to a model.
$\Sig_j$ can be viewed as a precursor to a Church-encoding for the theory of
signatures, but we only need contexts encoded in this way, and not other ToS
components. In functional programming, this representation is sometimes called
``finally tagless'' \cite{carette2007finally}, and it is used for defining and
interpreting embedded languages.
\end{definition}

In the following examples, we leave the abstracted $M : \ToS_{i,j}$ implicit.

\begin{example}{Infinitary constructors}. The universe $\U$ is closed under
the $\Piinf$ function type, which allows such functions to appear in the domains
of $\Pii$ types. This allows, for example, a signature for trees branching with
arbitrary small types. This is a signature at level 1, since we have $\Set_0$ as
a $\Pie$ domain type.
\begin{alignat*}{3}
& \mathsf{TreeSig} := \\
& \hspace{1em}\emptycon \ext (\mi{Tree} : \U)\\
& \hspace{1.35em}\ext (\mi{node} : \Pie\,\Set_0\,(\lambda A.\,(A\,\arrinf\,\mi{Tree})\arri\El\,\mi{Tree}))
\end{alignat*}
\end{example}

\begin{example}{Recursive equations}. Again, the universe is closed under $\Id$, which allows us to write equations
in $\Pii$ domains. A minimal (and trivial) example:
\begin{alignat*}{3}
  & \mathsf{RecEqSig} :=\\
  & \hspace{0.5em}\emptycon\,\ext (A : \U) \ext (a : \El\,A)\ext (f : \Pii\,(x : A)\,(\Id\,A\,x\,a\arri\El\,A))
\end{alignat*}
More interesting (and complicated) examples for recursive equations are boundary
conditions in various cubical type theories \cite{cchm,
  angiuli2016computational, angiuli2018cartesian}. Note that our $\Id$ allows
iterated equations as well, but these are all trivial in the semantics, where we
assume UIP.
\end{example}

\emph{Remark.} Since signatures are parametrized by a single universe level, all
external types in constructors must be contained in the same $\Set_j$
universe. We opted for this setup for the sake of simplicity. Cumulativity helps
here: it allows us to pick a $j$ level which is large enough to accommodate all
external types in a signature.

\section{Semantics}
\label{sec:semantics}

\subsection{Overview}\label{sec:overview}

For each signature, we would like to have at least
\begin{enumerate}
  \item A category of algebras, with homomorphisms as morphisms.
  \item A notion of induction, which requires a notion of dependent algebras.
  \item A proof that for algebras, initiality is equivalent to supporting induction.
\end{enumerate}

Following \cite{kaposi2019constructing}, we do this by creating a model of ToS,
where contexts are categories supporting the above requirements and
substitutions are appropriate structure-preserving functors. Then, each signature can be applied to this model, yielding an interpretation of the signature
as a structured category of algebras.

Our semantics has a ``type-theoretic'' flavor, which is inspired by the cubical
set model of Martin-Löf type theory by Bezem et al. \cite{cubicalmodel}. The core idea
is to avoid strictness issues by starting from basic ingredients which are already
strict enough. Hence, instead of modeling types as certain slices and
substitution by pullback, we model types as displayed categories with extra
structure, which naturally support strict reindexing.

We make a similar choice in the interpretation of signatures themselves: we use
structured CwFs instead of lex categories. The reason here is that CwFs allow us
to compute induction principles in strictly the same way as one would write in
type theory, since we have $\Ty$ and $\Tm$ for a primitive notion of dependent
objects and morphisms. In contrast, dependent objects in lex categories is a
derived notion, and the induction principles we get are only up to
isomorphism. This issue is perhaps not relevant from a purely categorical
perspective, but we are concerned with eventually implementing QIITs in proof
assistants, so we prefer if our semantics computes strictly. This was
demonstrated previously in \cite{hiit}, where we provided a program which
computed types of induction principles from signatures of higher
inductive-inductive types, and we believe that the same could be achieved for
the signatures and semantics described in this paper.

In the following, for given $i$ and $j$ levels, we define a model $\bM_{i,j} :
\ToS_{\max(i+1,j)+1, j}$ such that $\Con_{\bM_{i,j}}$ is a type of structured
categories (of algebras). The level $i$ marks the level of all internal sorts in
an algebra, and the level $j$ marks the level of all external sets in function
domains. Hence, every algebra has level $\max(i+1,j)$. The bump is only needed
for $i$, since algebras merely contain elements of $A : \Set_j$ types, while
inductive sets are themselves elements of $\Set_i$. For example, $\mi{NatAlg}_i :
\Set_{\max(i+1,0)} : \Set_{\max(i+1, 0)+1}$.

We present the components of the model in order. In the following, we use
\textbf{bold} font to disambiguate components of $\bM_{i,j}$ from components of other
structures. For example, we use $\boldsymbol{\sigma : \Sub\,\Gamma\,\Delta}$ to
denote a substitution in $\bM_{i,j}$.

The model involves a large amount of technical detail; we omit a significant
part of this, and only present the most salient parts.

\subsection{Contexts}

We define $\bCon : \Set_{\max(i+1,j)+1}$ as $\flCwF_{\max(i+1,j)}$.

\begin{definition}[Finite limit CwFs]\label{def:flCwF}
For each level $i$ we define $\flCwF_i : \Set_{i+1}$ as an iterated $\Sigma$-type
with the following components:
\begin{enumerate}
  \item A CwF with underlying sets all in $\Set_i$. We reuse the component
    notations $\Con$, $\Sub$, $\Ty$, etc.\ from Definition \ref{def:tos}.
  \item $\Sigma$-types $\Sg : (A : \Ty\,\Gamma)\ra \Ty\,(\Gamma \ext A)
    \ra \Ty\,\Gamma$, with term formers $\proj1$, $\proj2$ and $\blank,\blank$.
  \item Identity type $\Id : (A : \Ty\,\Gamma)\ra \Tm\,\Gamma\,A\ra
    \Tm\,\Gamma\,A\ra \Ty\,\Gamma$, with $\refl$, equality reflection and UIP.
  \item Constant families. This includes a type former $\Kfam : \Con \ra
    \Ty\,\Gamma$, where $\Gamma$ is implicitly quantified, together with $\lamK
    : \Sub\,\Gamma\,\Delta \ra \Tm\,\Gamma\,(\Kfam\,\Delta)$ and its inverse
    $\appK$. The idea is that $\Kfam\,\Delta$ is a representation of $\Delta$ as
    a type in any context. Clairambault and Dybjer called constant families
    ``democracy'' in \cite{clairambault2014biequivalence}.
\end{enumerate}
We abbreviate the additional structure on CwFs consisting of $\Sigma$, $\Id$ and
$\Kfam$ as \emph{fl-structure}.
\end{definition}

\begin{definition}[Notion of induction in an flCwF]\label{def:induction}
Given $\bGamma : \flCwF_i$, we have the following predicate on contexts:
\begin{alignat*}{3}
  & \Inductive : \Con_{\bGamma} \ra \Set_i \\
  & \Inductive\,\Gamma := (A : \Ty_{\boldsymbol{\Gamma}}\,\Gamma)\ra \Tm_{\bGamma}\,\Gamma\,A
\end{alignat*}
\end{definition}

For an example, if we interpret $\mathsf{\NatSig}$ in the $\bM$ model, we get an
flCwF of natural number algebras, where $\Con$ is the type of algebras and
$\Sub\,\Gamma\,\Delta$ is the type of homomorphisms between $\Gamma$ and
$\Delta$ algebras. $\Ty$ is the type of displayed algebras, and $\Tm$ is the
type of their sections:
\begin{alignat*}{3}
  & \Ty\,(N,\,z,\,s) \equiv(N^D : N \to \Set)\\ & \hspace{2em}\times (N^D\,z)
  \times ((n : N)\to N^D\,n \to N^D(s\,n))\\ &
  \Tm\,(N,\,z,\,s)\,(N^D,\,z^D,\,s^D) \equiv (N^S : (n : N)\to
  N^D\,n)\\ & \hspace{2em}\times (N^S\,z = z^D)\times ((n : N) \to N^S\,(s\,n) =
  s^D\,n\,(N^S\,n))
\end{alignat*}
Thus, for natural number algebras, $\Inductive$ is exactly the predicate which
holds when an algebra supports induction.


\begin{theorem}[Equivalence of initiality and induction, c.f. \cite{kaposi2019constructing}]
\label{thm:initialind} An object $\Gamma : \Con_{\bGamma}$ supports induction if and only if it is
initial. Moreover, induction and initiality are both proof-irrelevant
predicates. \qed
\end{theorem}

The reason for the ``finite limit CwF'' naming is the following: Clairambault
and Dybjer showed that the 2-category of flCwFs is biequivalent to the
2-category of finitely complete categories
\cite{clairambault2014biequivalence}. In particular, in an flCwF the categorical
product of $\Gamma$ and $\Delta$ can be given as $\Gamma \ext \Kfam\, \Delta$,
and the equalizer of $\sigma$ and $\delta$ as $\Gamma\ext
\Id\,(\K\,\Delta)\,(\lamK\,\sigma)\,(\lamK\,\delta)$. While showing equivalence
of initiality and induction does not need all flCwF components (e.g.\ $\Sigma$
is not needed), we build the full flCwF semantics in order to connect to
Clairambault's and Dybjer's results.

In order to talk about weak structure-preservation in the interpretation of
substitutions, we need to specify isomorphisms for contexts and types.

\begin{definition}
A \emph{context isomorphism} is an invertible morphism $\sigma :
\Sub\,\Gamma\,\Delta$. We note the inverse as $\sigma^{-1}$. We also use the
notation $\sigma : \Gamma \simeq \Delta$.
\end{definition}

\begin{definition}[Type categories, c.f.\ \cite{clairambault2014biequivalence}]
\label{def:type_categories} For each $\Gamma : \Con$, there is a category
whose objects are types $A : \Ty\,\Gamma$, and morphisms from $A$ to $B$ are
terms $t : \Tm\,(\Gamma\,\ext\,A)\,(B[\p])$. Identity morphisms are given by $\q
: \Tm\,(\Gamma\,\ext\,A)\,(A[\p])$, and composition $t \circ u$ by $t[\p,
  u]$. The assignment of type categories to contexts extends to a split indexed
category. For each $\sigma : \Sub\,\Gamma\,\Delta$, there is a functor from
$\Ty\,\Delta$ to $\Ty\,\Gamma$, which sends $A$ to $A[\sigma]$ and $t :
\Tm\,(\Gamma\,\ext\,A)\,(B[\p])$ to $t[\sigma\circ \p, \q]$.
\end{definition}

\begin{definition} A \emph{type isomorphism}, notated $t : A \simeq B$ is an isomorphism in a type category. We note the inverse as $t^{-1}$.
\end{definition}

\subsection{Substitutions}
\label{sec:substitutions}

A \emph{weak flCwF morphism} $\boldsymbol{\sigma : \Sub\,\Gamma\,\Delta}$ is a
functor between underlying categories, which also maps types to types and terms
to terms, and satisfies the following mere properties:
  \begin{enumerate}
    \item $\bsigma\,(A[\sigma]) = (\bsigma\,A)\,[\bsigma\,\sigma]$
    \item $\bsigma\,(t[\sigma]) = (\bsigma\,t)\,[\bsigma\,\sigma]$
    \item The unique map $\epsilon : \Sub\,(\bsigma\,\emptycon)\,\emptycon$ has a retraction.
    \item Each $(\bsigma\,\p,\,\bsigma\,\q) : \Sub\,(\bsigma\,(\Gamma\,\ext\,A))\,(\bsigma\,\Gamma\,\ext\,\bsigma\,A)$ has an inverse.
  \end{enumerate}

In short, $\bsigma$ preserves substitution strictly and preserves empty context
and context extension up to isomorphism. We notate the evident isomorphisms as
$\bsigma_{\emptycon} : \bsigma\,\emptycon \simeq \emptycon$ and $\bsigma_{\ext}
: \bsigma\,(\Gamma\,\ext\,A)\,\simeq\,\bsigma\,\Gamma\,\ext\,\bsigma\,A$. Our
notion of weak morphism is the same as in \cite{birkedal2020modal}, when
restricted to CwFs.

Note that the definition we just gave lives in $\Set_{\max(i+1,j)}$, but by
cumulativity it is also in $\Set_{\max(i+1,j)+1}$, as required by our $\bM_{i,j}
: \ToS_{\max(i+1,j)+1, j}$ specification of the model being defined.

\begin{theorem}\label{thm:flpres}
Every $\boldsymbol{\sigma : \Sub\,\Gamma\,\Delta}$ preserves fl-structure up to
type isomorphism. That is, we have
\begin{alignat*}{3}
  & \bsigma_{\Sigma} : \bsigma\,(\Sigma\,A\,B) \simeq \Sigma\,(\bsigma\,A)\,((\bsigma\,B)[\bsigma_{\ext}^{-1}]) \\
  & \bsigma_{\K} : \bsigma\,(\K\,\Delta) \simeq \K\,(\bsigma\,\Delta) \\
  & \bsigma_{\Id} : \bsigma\,(\Id\,t\,u) \simeq \Id\,(\bsigma\,t)\,(\bsigma\,u)
\end{alignat*}
These are all natural in the following sense: for $\sigma :
\Sub_{\bGamma}\,\Gamma\,\Delta$, the functorial action of $\bsigma\,\sigma :
\Sub_{\bDelta}\,(\bsigma\,\Gamma)\,(\bsigma\,\Delta)$ on $\bsigma_{\Sigma}$ (in
the $\bsigma\,\Gamma$ context) is equal to $\bsigma_{\Sigma}$ (in
$\bsigma\,\Delta$), and similarly for $\bsigma_{\K}$ and $\bsigma_{\Id}$.

Moreover, $\bsigma$ preserves all term and substitution formers in the
fl-structure. For example, $\bsigma\,(\proj1\,t) = \proj1\,
(\bsigma_{\Sigma}[\id, \bsigma\,t])$.
\end{theorem}
\begin{proof}
For $\bsigma_{\Sigma}$, we construct the following context isomorphism:
\begin{alignat*}{3}
& (\bsigma\,\Gamma\,\ext\,\bsigma\,(\Sigma\,A\,B)) \simeq
  (\bsigma\,\Gamma\,\ext\,\bsigma\,A\,\ext\,(\bsigma\,B)[\bsigma_{\ext}^{-1}]) \\
& \simeq (\bsigma\,\Gamma\,\ext\,\Sigma\,(\bsigma\,A)\,((\bsigma\,B)[\bsigma_{\ext}^{-1}]))
\end{alignat*}
This isomorphism is the identity on $\bsigma\,\Gamma$, hence we can extract the
desired $\bsigma_{\Sigma} : \bsigma\,(\Sigma\,A\,B) \simeq
\Sigma\,(\bsigma\,A)\,((\bsigma\,B)[\bsigma_{\ext}^{-1}])$ from it.

For $\bsigma_{\K}$, note the following:
\begin{alignat*}{3}
  & (\emptycon\,\ext\,\bsigma\,(\K\,\Delta)) \simeq
    (\bsigma\,\emptycon\,\ext\,\bsigma\,(\K\,\Delta)) \simeq
    \bsigma\,(\emptycon\,\ext\,\K\,\Delta)\\
  & \simeq \bsigma\,\Delta \simeq (\emptycon\,\ext\,K\,(\bsigma\,\Delta))
\end{alignat*}
This yields a type isomorphism $\bsigma\,(\K\,\Delta) \simeq
\K\,(\bsigma\,\Delta)$ in the empty context, and we use the functorial action of
$\epsilon : \Sub\,\Gamma\,\emptycon$ to weaken it to any $\Gamma$ context.

For $\bsigma_{\Id}$, both component morphisms can be constructed by $\refl$ and
equality reflection, and the morphisms are inverses by UIP. We omit here the
verification of naturality and that $\bsigma$ preserves term and substitution
formers in the fl-structure.
\end{proof}

\subsection{Identity and Composition}
\label{sec:idcomp}

$\bid : \bSub\,\bGamma\,\bGamma$ is defined in the obvious way, with identities for
underlying functions and for preservation morphisms.

For $\boldsymbol{\sigma \circ \delta}$, the underlying functions are given by
function composition, and the preservation morphisms are given as follows:
\begin{alignat*}{3}
  & (\boldsymbol{\sigma \circ \delta})_{\emptycon}^{-1} :=
    \bsigma\,\bdelta_{\emptycon}^{-1} \circ \bdelta_{\emptycon}^{-1} \\
  & (\boldsymbol{\sigma \circ \delta})_{\ext}^{-1} :=
    \bsigma\,\bdelta_{\ext}^{-1} \circ \bdelta_{\ext}^{-1}
\end{alignat*}

It is easy to verify the left and right identity laws and associativity for
$\boldsymbol{\blank\circ\blank}$.

\begin{lemma}\label{lem:idcomppres}
The derived preservation isomorphisms for the fl-structure can be decomposed
analogously; all derived isomorphisms in $\bid$ are identities, and we have
\begin{alignat*}{3}
  & (\boldsymbol{\sigma \circ \delta})_{\Sigma} =
  \bsigma\,\bdelta_{\Sigma} \circ \bdelta_{\Sigma}\\
  & (\boldsymbol{\sigma \circ \delta})_{\K} =
  \bsigma\,\bdelta_{\K} \circ \bdelta_{\K}\\
  & (\boldsymbol{\sigma \circ \delta})_{\Id} =
  \bsigma\,\bdelta_{\Id} \circ \bdelta_{\Id}
\end{alignat*}
On the right sides, $\blank\circ\blank$ refers to composition of type morphisms.
\end{lemma}
\begin{proof}
In the case of $\Id$, the equations hold immediately by UIP. For $\Sigma$ and
$\K$, we prove by flCwF computation and straightforward unfolding of
definitions.
\end{proof}

\subsection{Empty Context}
The empty context $\boldsymbol{\emptycon : \Con}$ is the terminal flCwF, which
has all underlying sets defined as $\top$ (or constantly $\top$), with an
evident unique $\boldsymbol{\epsilon : \Sub\,\Gamma\,\emptycon}$. Since
$\bepsilon$ is a strict flCwF morphism, $\bepsilon_{\emptycon}^{-1}$ and
$\bepsilon_{\ext}^{-1}$ are both identity morphisms.

\subsection{Types}

We define $\boldsymbol{\Ty\,\Gamma} : \Set_{\max(i+1, j)+1}$ as the type of
split flCwF-isofibrations over $\bGamma$, at level $\max(i+1, j)$. We extend
Ahrens' and Lumsdaine's displayed categories and their definition of
isofibrations \cite{displayedcats}.  We first define displayed flCwFs, then
specify iso-cleaving as additional structure on top of that.

\begin{definition}[Displayed flCwF]
The type of displayed flCwFs at level $i$ is given as the logical predicate
interpretation (see e.g.\ \cite{bernardy12parametricity} or \cite{hiit}) of
$\flCwF_i$. For each flCwF component in $\bGamma$, there is a component in a
displayed flCwF which ``lies over'' it.

\emph{Notation.} In situations where we need to refer to both ``base'' and
displayed things, we give \ul{underlined} names to contexts, substitutions,
types and terms in a base flCwF. For example, we may have $\ulGamma :
\Con_{\bGamma}$ living in $\boldsymbol{\Gamma : \Con}$, and $\Gamma :
\Con_{\bA}\,\ulGamma$ living in a displayed flCwF over $\bGamma$. We only use
underlining on cETT variable names, and overload flCwF component names for
displayed counterparts. For example, a $\Con$ component is named the same in
a base flCwF and a displayed one.

Concretely, a displayed flCwF $\bA$ over $\bGamma$ has the following underlying
sets, which we call displayed contexts, substitutions, types and terms
respectively.
\begin{alignat*}{3}
  & \Con_{\bA} && : \Con_{\bGamma}\ra \Set_i\\
  & \Sub_{\bA} && : \Con_{\bA}\,\ulGamma \ra \Con_{\bA}\,\ulDelta \ra \Sub_{\bGamma}\,\ulGamma\,\ulDelta \ra \Set_i \\
  & \Ty_{\bA}  && : \Con_{\bA}\,\ulGamma \ra \Ty_{\bGamma}\,\ulGamma \ra \Set_i\\
  & \Tm_{\bA}  && : (\Gamma : \Con_{\bA}\,\ulGamma)\ra \Ty_{\bA}\,\Gamma\,\ulA \ra \Tm_{\bGamma}\,\ulGamma\,\ulA \ra \Set_i
\end{alignat*}

Above, we implicitly quantify over $\ulGamma$, $\ulDelta$ and $\ulA$ base
parameters. We also have the following components for empty context, context
extension and substitution. We omit listing other components here.
\begin{alignat*}{3}
  & \emptycon_{\bA} && : \Con_{\bA}\,\emptycon_{\bGamma}\\
  & \ext_{\bA}      && : (\Gamma : \Con_{\bA}\,\ulGamma)\ra \Ty_{\bA}\,\Gamma\,\ulA \ra
                     \Con_{\bA}\,\Gamma\,(\ulGamma \ext_{\bGamma} \ulA)\\
  & \blank[\blank]_{\bA} && : \Ty_{\bA}\,\Delta\,\ulA \ra \Sub_{\bA}\,\Gamma\,\Delta\,\ulsigma
                     \ra \Ty_{\bA}\,\Gamma\, (\ulA[\ulsigma]_{\bGamma})\\
  & \blank[\blank]_{\bA} && : \Tm_{\bA}\,\Delta\,A\,\ult \ra (\sigma : \Sub_{\bA}\,\Gamma\,\Delta\,\ulsigma)\\
  & && \hspace{0.5em}\ra \Tm_{\bA}\,\Gamma\, (A[\sigma]_{\bA})\,(\ult[\ulsigma]_{\bGamma})
\end{alignat*}
\end{definition}

In the following we will often omit $_{\bGamma}$ and $_{\bA}$ subscripts on
components; for example, in the type $\Con_{\bA}\,\emptycon$, the $\emptycon$ is
clearly a base component in $\bGamma$.

We also need displayed counterparts to the previously defined derived notions on
flCwFs; these are again given as logical predicate interpretations of the
non-displayed definitions.

\begin{definition}[Displayed type categories]
For each $\Gamma : \Con_{\bA}\,\ulGamma$, there is a displayed category over the
type category $\Ty_{\bGamma}\,\ulGamma$, whose objects over $\ulA :
\Ty_{\bGamma}\,\ulGamma$ are elements of $\Ty_{\bA}\,\Gamma\,\ulA$, and
displayed morphisms over $\ult : \Tm_{\bGamma}\,(\ulGamma \ext
\ulA)\,(\ulB[\p])$ are elements of $\Tm_{\bA}\,(\Gamma \ext
A)\,(B[\p])\,\ult$. The identity morphism is given by $\q_{\bA}$, and the
composition of $t$ and $u$ is $t[\p_{\bA},u]$. Analogously to Definition
\ref{def:type_categories}, this extends to a displayed split indexed category.
\end{definition}

\begin{definition}[Displayed isomorphisms]
A \emph{displayed context isomorphism} over $\ulsigma : \ulGamma \simeq
\ulDelta$, notated $\sigma : \Gamma \simeq_{\ulsigma} \Delta$, is an invertible
displayed morphism $\sigma : \Sub_{\bA}\,\Gamma\,\Delta\,\ulsigma$, with inverse
$\sigma^{-1} : \Sub_{\bA}\,\Delta\,\Gamma\,\ulsigma^{-1}$. A \emph{displayed
  type isomorphism} over $\ult : \ulA \simeq \ulB$, notated $t : A \simeq_{\ult}
B$, is an isomorphism in a displayed type category.
\end{definition}

\begin{definition}
A \emph{vertical morphism} lies over an identity morphism. We use this
definition for context morphisms (substitutions) and type morphisms as well.
\end{definition}

In contrast to \cite{kaposi2019constructing}, it is not sufficient to model
types as displayed flCwFs. In ibid.\ the universe $\U$ in ToS was empty, and all
substitutions were ``neutral'', i.e.\ semantic subsitutions were functors which
may permute, duplicate or forget components of algebras, or freely reinterpret
components, and it is easy to see that all such functors strictly preserve
limits. In contrast, the current $\U$ is not empty: it is closed under identity
and infinitary function types. Hence, substitutions and terms are not neutral
anymore, as they can contain canonical type codes in $\U$. Semantically, these
canonical type codes do not merely reshuffle structure, hence they preserve
limits only weakly. We will return to this in Section
\ref{sec:infinitaryfunction}. We are forced to use a weaker semantics where
fl-structure is not preserved strictly, and we also need to add additional
structure to displayed flCwFs which expresses preservation of base isomorphisms.

\begin{definition}[Context iso-cleaving] This lifts a base context isomorphism to a displayed one. It consists of
\begin{alignat*}{3}
  & \coe &&: \ulGamma \simeq \ulDelta \ra \Con_{\bA}\,\ulGamma \ra \Con_{\bA}\,\ulDelta\\
  & \coh &&: (\ulsigma : \ulGamma \simeq \ulDelta)(\Gamma : \Con_{\bA}\,\ulGamma)
           \ra \Gamma \simeq_{\ulsigma} \coe\,\ulsigma\,\Gamma\\
  & \coe^{\id} && : \coe\,\id\,\Gamma = \Gamma\\
  & \coe^{\circ} && : \coe\,(\ulsigma\circ\uldelta)\,\Gamma = \coe\,\ulsigma\,(\coe\,\uldelta\,\Gamma)\\
  & \coh^{\id} && : \coh\,\id\,\Gamma = \id\\
  & \coh^{\circ} && : \coh\,(\ulsigma\circ\uldelta)\,\Gamma = \coh\,\ulsigma\,(\coe\,\uldelta\,\Gamma)
          \circ \coh\,\uldelta\,\Gamma
\end{alignat*}
Here, $\coe$ and $\coh$ abbreviate ``coercion'' and ``coherence'' respectively.
\end{definition}

\begin{definition}[Type iso-cleaving] This consists of
\begin{alignat*}{3}
  & \coe &&: \ulA \simeq \ulB \ra \Ty_{\bA}\,\Gamma\,\ulA \ra \Ty_{\bA}\,\Gamma\,\ulB\\
  & \coh &&: (\ult : \ulA \simeq \ulB)(A : \Ty_{\bA}\,\Gamma\,\ulA)
           \ra A \simeq_{\ult} \coe\,\ult\,A\\
  & \coe^{\id} && : \coe\,\id\,A = A\\
  & \coe^{\circ} && : \coe\,\ult\,(\coe\,\uldelta\,A) = \coe\,(\ult\circ\uldelta)\,A\\
  & \coh^{\id} &&: \coh\,\id\,A = \id\\
  & \coh^{\circ} &&: \coh\,(\ult\circ\uldelta)\,A = \coh\,\ult\,(\coe\,\uldelta\,A)
          \circ \coh\,\uldelta\,A
\end{alignat*}
Additionally, for $\sigma : \Sub_{\bA}\,\Gamma\,\Delta\,\ulsigma$, we have
\begin{alignat*}{3}
  & \coe[] &&: \coe\,(\ult[\ulsigma])\,(A[\sigma]) = (\coe\,\ult\,A)[\sigma]\\
  & \coh[] &&: \coh\,(\ult[\ulsigma\circ \p,\q])\,(A[\sigma]) = (\coh\,\ult\,A)[\sigma]
\end{alignat*}

\end{definition}

\begin{definition} A \emph{split flCwF isofibration} is a displayed flCwF equipped with iso-cleaving for contexts and types.
\end{definition}

\emph{Remark.} It is not possible to model types as fibrations or opfibrations,
because we have no restriction on the variance of ToS types. For example, the
type which extends a pointed set to a natural number signature, is neither a
fibration nor an opfibration.

\subsection{Type Substitution}
We aim to define $\boldsymbol{\blank[\blank] : \Ty\,\Delta \ra
  \Sub\,\Gamma\,\Delta \ra \Ty\,\Gamma}$, such that $\boldsymbol{A[\id]} = \bA$
and $\boldsymbol{A[\sigma\circ\delta]} = \boldsymbol{A[\sigma][\delta]}$. The
underlying sets are given by simple composition:
\begin{alignat*}{3}
  & \Con_{\boldsymbol{A[\sigma]}}\,\ulGamma && := \Con_{\bA}\,(\bsigma\,\ulGamma)\\
  & \Sub_{\boldsymbol{A[\sigma]}}\,\Gamma\,\Delta\,\ulsigma && :=
    \Sub_{\bA}\,\Gamma\,\Delta\,(\bsigma\,\ulsigma)\\
  & \Ty_{\boldsymbol{A[\sigma]}}\,\Gamma\,\ulA && :=
      \Ty_{\bA}\,\Gamma\,(\bsigma\,\ulA)\\
  & \Tm_{\boldsymbol{A[\sigma]}}\,\Gamma\,A\,\ult && :=
      \Tm_{\bA}\,\Gamma\,A\,(\bsigma\,\ult)
\end{alignat*}

Moreover, $\id_{\boldsymbol{A[\sigma]}} := \id_{\bA}$,
$\sigma\circ_{\boldsymbol{A[\sigma]}}\delta := \sigma\circ_{\bA}\delta$, and
likewise components for substitution are given by corresponding components in
$\bA$. Context and type formers are given by coercing $\bA$
structures along $\bsigma$ preservation isomorphisms. For example:
\begin{alignat*}{3}
  &\emptycon_{\boldsymbol{A[\sigma]}} && :=
    \coe\,\bsigma_{\emptycon}^{-1}\,\emptycon_{\bA}\\
  &\Gamma\ext_{\boldsymbol{A[\sigma]}}A && :=
    \coe\,\bsigma_{\ext}^{-1}\,(\Gamma\ext_{\bA} A)\\
  &\Id_{\boldsymbol{A[\sigma]}}\,t\,u && :=
    \coe\,\bsigma_{\Id}^{-1}\,(\Id_{\bA}\,t\,u)
\end{alignat*}

Term and substitution formers are given by composing $\coh$-lifted
isomorphisms with term and substitution formers from $\bA$. For example:
\begin{alignat*}{3}
  & \epsilon_{\boldsymbol{A[\sigma]}} && :=
    \coh\,\bsigma_{\emptycon}^{-1}\,\emptycon_{\bA} \circ \epsilon_{\bA}\\
  & \p_{\boldsymbol{A[\sigma]}} && :=
    \p_{\bA} \circ (\coh\,\bsigma_{\ext}^{-1}\,(\Gamma\ext A))^{-1}\\
  & \appK_{\boldsymbol{A[\sigma]}}\,t && :=
    \appK_{\bA}\,((\coh\,\bsigma_{\K}\,(\K\,\Delta))^{-1}\circ t)
\end{alignat*}
Equations for term and type substitution follow from naturality of preservation
isomorphisms in $\bsigma$, $\coe[]$, $\coh[]$ and substitution equations in
$\bA$.

Iso-cleaving is given by iso-cleaving in $\bA$ and the action of $\bsigma$ on
isomorphisms, e.g.\ we have $\coe_{\boldsymbol{A[\sigma]}}\,\ulsigma\,\Gamma
:= \coe_{\bA}\,(\bsigma\,\ulsigma)\,\Gamma$.

Functoriality of type substitution, i.e.\ $\boldsymbol{A[\id]} = \bA$ and
$\boldsymbol{A[\sigma\circ\delta]} = \boldsymbol{A[\sigma][\delta]}$, follows
from Lemma \ref{lem:idcomppres} and split cleaving given by $\coe^{\id}$,
$\coe^{\circ}$, $\coh^{\id}$ and $\coh^{\circ}$ laws in $\bA$.

\subsection{Terms}

$\boldsymbol{\Tm\,\Gamma\,A} : \Set_{\max(i+1, j)+1}$ is defined as the type of
\emph{weak flCwF sections} of $\bA$. The underlying functions of $\bt :
\bTm\,\bGamma\,\bA$ are as follows:
\begin{alignat*}{3}
  & \bt : (\ulGamma : \Con_{\bGamma}) \ra \Con_{\bA}\,\ulGamma\\
  & \bt : (\ulsigma : \Sub_{\bGamma}\,\ulGamma\,\ulDelta)
         \ra \Sub_{\bA}\,(\bt\,\ulGamma)\,(\bt\,\ulDelta)\,\ulsigma\\
  & \bt : (\ulA : \Ty_{\bGamma}) \ra \Ty_{\bA}\,(\bt\,\ulGamma)\,\ulA\\
  & \bt : (\ult : \Tm_{\bGamma}\,\ulGamma\,\ulA) \ra
          \Tm_{\bA}\,(\bt\,\ulGamma)\,(\bt\,\ulA)\,\ult
\end{alignat*}
Such that
\begin{enumerate}
  \item $\bt\,(\ulA[\ulsigma]) = (\bt\,\ulA)\,[\bt\,\ulsigma]$
  \item $\bt\,(\ult[\ulsigma]) = (\bt\,\ult)\,[\bt\,\ulsigma]$
  \item The unique map $\epsilon_{\bA} : \Sub\,(\bt\,\emptycon)\,\emptycon\,\id$ has a vertical retraction.
  \item Each $(\bt\,\p,\,\bt\,\q) : \Sub\,(\bt\,(\ulGamma\,\ext\,\ulA))\,(\bt\,\ulGamma\,\ext\,\bt\,\ulA)\,\id$ has a vertical inverse.
\end{enumerate}

Similarly to Section \ref{sec:substitutions}, we denote the evident preservation
isomorphisms as $\bt_{\emptycon} : \bt\,\emptycon \simeq_{\id} \emptycon$ and
$\bt_{\ext} : \bt\,(\ulGamma\ext \ulA) \simeq_{\id} \bt\,\ulGamma \ext
\bt\,\ulA$. In short, weak section is a dependently typed analogue of weak
morphism, with dependent underlying functions and displayed preservation
isomorphisms. We also have the derived fl-preservation isomorphisms.

\begin{theorem} A weak section $\boldsymbol{t : \Tm\,\Gamma\,A}$ preserves fl-structure up to vertical type isomorphisms, that is, the following are derivable:
\begin{alignat*}{3}
  & \bt_{\Sigma} : \bt\,(\Sigma\,\ulA\,\ulB) \simeq_{\id} \Sigma\,(\bt\,\ulA)\,((\bt\,\ulB)[\bt_{\ext}^{-1}]) \\
  & \bt_{\K} : \bt\,(\K\,\ulDelta) \simeq_{\id} \K\,(\bt\,\ulDelta) \\
  & \bt_{\Id} : \bt\,(\Id\,\ult\,\ulu) \simeq_{\id} \Id\,(\bt\,\ult)\,(\bt\,\ulu)
\end{alignat*}
Also, the above isomorphisms are natural in the sense of Theorem
\ref{thm:flpres}, and $\bt$ preserves type and substitution formers in the
fl-structure.
\end{theorem}
\begin{proof}
The construction of isomorphisms is the same as in Theorem
\ref{thm:flpres}. Indeed, every construction there has a displayed counterpart
which we can use here.
\end{proof}

We note though that the move from Theorem \ref{thm:flpres} to here is not simply a
logical predicate translation, because we are only lifting the codomain of a
weak morphism to a displayed version, and we leave the domain non-displayed. We
leave to future work the investigation of such asymmetrical (or
``modal'') logical predicate translations.

\subsection{Term Substitution}
$\boldsymbol{\blank[\blank] : \Tm\,\Delta\,A \ra (\sigma : \Sub\,\Gamma\,\Delta)
  \ra \Tm\,\Gamma\,(A[\sigma])}$ is given similarly to
$\boldsymbol{\blank\circ\blank}$ in Section \ref{sec:idcomp}. Underlying functions
are given by function composition, and preservation morphisms are also similar:
\begin{alignat*}{3}
  & (\boldsymbol{t[\sigma]})_{\emptycon}^{-1} :=
    \bt\,\bsigma_{\emptycon}^{-1} \circ \bt_{\emptycon}^{-1} \\
  & (\boldsymbol{t[\sigma]})_{\ext}^{-1} :=
    \bt\,\bsigma_{\ext}^{-1} \circ \bt_{\ext}^{-1}
\end{alignat*}
We also have the same decomposition of derived isomorphisms as in Lemma
\ref{lem:idcomppres}. We do not have to show functoriality of term substitution
here, since that is derivable in any CwF, see e.g. \cite{kaposi2019constructing}.

\subsection{Context Extension and Comprehension}

$\boldsymbol{\Gamma \ext A : \Con}$ is defined as the \emph{total flCwF} of
$\bA$. This is given by bundling together all displayed flCwF components in
$\bA$ with corresponding base components in $\bGamma$, using the metatheoretic
$\Sigma$-type. It is a straightforward extension of total categories in
\cite{displayedcats}.

$\boldsymbol{\p : \Sub\,(\Gamma\ext A)\,\Gamma}$ is a strict morphism given by
taking a first projection for each component. $\boldsymbol{\q : \Tm\,(\Gamma\ext
  A)\,(A[\p])}$ is likewise a strict flCwF section given by second
projections. Substitution extension $\boldsymbol{\sigma,\,t}$ is given by
pointwise combining $\bsigma$ and $\bt$ with metatheoretic $\Sigma$ pairing,
e.g.\ $\Con_{\boldsymbol{(\sigma,t)}}\,\ulGamma :=
(\bsigma\,\ulGamma,\,\bt\,\ulGamma)$.

\subsection{Universe}
\label{sec:universe}

\begin{definition}
For a level $i$, we write $\bSet_i$ for the flCwF of sets where $\Con_{\bSet_i}
:= \Set_i$ and $\Sub_{\bSet_i}\,\Gamma\,\Delta := \Gamma \ra \Delta$.
\end{definition}

We define $\boldsymbol{\U : \Ty\,\Gamma}$ as the isofibration which is
constantly $\bSet_i$. A constant isofibration does not actually depend on the
base flCwF, and has trivial iso-cleaving where $\coe$-s are identity
functions. Hence, we have $\Con_{\bU}\,\ulGamma := \Set_i$ and
$\Sub_{\bU}\,\Gamma\,\Delta\,\ulsigma := \Gamma \ra \Delta$.

\emph{Remark.} The type $\bTm\,\bGamma\,\bU$ is strictly equal to
$\bSub\,\bGamma\,\bSet_i$, so it is helpful to think about semantic elements of
the universe as weak morphisms from $\bGamma$ to $\bSet_i$.

\subsection{Elements of the Universe}

We define $\boldsymbol{\El : \Tm\,\Gamma\,\U \ra \Ty\,\Gamma}$ as discrete
isofibration formation. For $\boldsymbol{a : \Tm\,\Gamma\,\U}$, the underlying
sets of $\bEl\,\ba$ are the following:
\begin{alignat*}{3}
  & \Con_{\bEl\,\ba}\,\ulGamma && := \ba\,\ulGamma\\
  & \Sub_{\bEl\,\ba}\,\Gamma\,\Delta\,\ulsigma && := \ba\,\ulsigma\,\Gamma = \Delta\\
  & \Ty_{\bEl\,\ba}\,\Gamma\,\ulA && := \ba\,\ulA\,\Gamma\\
  & \Tm_{\bEl\,\ba}\,\Gamma\,A\,\ult && := \ba\,\ult\,\Gamma = A
\end{alignat*}
Hence, in $\bEl\,\ba$, $\Sub$ and $\Tm$ are propositional. We use the
isomorphisms $\ba_{\emptycon} : \ba\,\emptycon \simeq \top$ and $\ba_{\ext} :
\ba\,(\ulGamma\ext\ulA) \simeq (\Gamma :
\ba\,\ulGamma)\times(\ba\,\ulA\,\Gamma)$ to define empty context and context
extension:
\begin{alignat*}{3}
  & \emptycon_{\bEl\,\ba} && := \ba_{\emptycon}^{-1}\,\tt\\
  & (\Gamma\ext_{\bEl\,\ba} A) && := \ba_{\ext}^{-1}\,(\Gamma,\,A)
\end{alignat*}
We likewise use preservation isomorphisms to define $\K$, $\Id$ and $\Sigma$.
Context coercion is $\coe\,\ulsigma\,\Gamma := \ba\,\ulsigma\,\Gamma$. Type
coercion, for $A : \ba\,\ulA\,\Gamma$ is given as $\coe\,\ult\,A :=
\ba\,\ult\,(\ba_{\ext}^{-1}\,(\Gamma,\,A))$.

\subsection{Inductive Function Space}

For $\boldsymbol{a : \Tm\,\Gamma\,\U}$ and $\boldsymbol{B :
  \Ty\,(\Gamma\ext\El\,a)}$, we aim to define $\boldsymbol{\Pi\,a\,B}
\boldsymbol{:} \bTy\,\bGamma$. We define this as a dependent product of
isofibrations, indexed by a discrete domain. The discreteness is essential: with
a general $\bA \boldsymbol{:} \bTy\,\bGamma$ domain, $\bPii$ would not be
definable because of variance issues. Indeed, the category of categories is not
locally cartesian closed and does not support a general $\Pi$ type
\cite[Section A1.5]{johnstone2002sketches}.

Contexts are products of $\bB$-contexts, and types are products of $\bB$-types,
indexed respectively by contexts and types of $\bEl\,\ba$.
\begin{alignat*}{3}
  & \Con_{(\bPii\,\ba\,\bB)}\,\ulGamma &&:= (\gamma : \ba\,\ulGamma)\ra \Con_{\bB}\,(\ulGamma, \gamma)\\
  & \Ty_{(\bPii\,\ba\,\bB)}\,\Gamma\,\ulA &&:= (\gamma : \ba\,\ulGamma)(a : \ba\,\ulA\,\gamma)\ra \Ty_{\bB}\,(\Gamma\,\gamma)\,(\ulA, a)
\end{alignat*}
Note that since $\bB$ is over the total $\boldsymbol{(\Gamma\ext\El\,a)}$,
$\Con_{\bB}$ has a $\Sigma$-typed argument, and likewise the last argument
of every $\bB$ component.  We could define substitutions similarly, as
products of substitutions:
\begin{alignat*}{3}
  & \Sub_{(\bPii\,\ba\,\bB)}\,\Gamma\,\Delta\,\ulsigma :=
  (\gamma : \ba\,\ulGamma)(\delta : \ba\,\ulDelta)
  (\sigma : \Sub_{(\bEl\,\ba)}\,\gamma\,\delta\,\ulsigma)\\
  &\hspace{6.5em}\ra \Sub_{\bB}\,(\Gamma\,\gamma)\,(\Delta\,\delta)\,(\ulsigma, \sigma)
\end{alignat*}
This would work, but we know that $\Sub_{(\bEl\,\ba)}\,\gamma\,\delta\,\ulsigma$ is defined
as $\ba\,\ulsigma\,\gamma = \delta$, so we can eliminate $\sigma$ by singleton contraction,
and use the following equivalent definition:
\begin{alignat*}{3}
  & \Sub_{(\bPii\,\ba\,\bB)}\,\Gamma\,\Delta\,\ulsigma :=
  (\gamma : \ba\,\ulGamma)\ra\Sub_{\bB}\,(\Gamma\,\gamma)\,(\Delta\,(\ba\,\ulsigma\,\gamma)\,(\ulsigma, \refl)
\end{alignat*}
The benefit of the contracted definition is that it computes preservation laws
in algebra homomorphisms strictly as expected, while the non-contracted
definition computes homomorphisms as functional logical relations.

Terms are also given as a singleton-contracted version of products of terms. In
$\bPii\,\ba\,\bB$, all other structure is given pointwise by $\bB$-structure.

Iso-cleaving is given by transporting indices backwards in $\bEl\,\ba$ and outputs forwards
in $\bB$:
\begin{alignat*}{3}
  & \coe\,\ulsigma\,\Gamma &&:=
    \lambda\,\gamma.\,\coe_{\bB}\,(\ulsigma,\refl)\,(\Gamma\,(\ba\,(\ulsigma^{-1})\,\gamma))\\
  & \coe\,\ult\,A &&:=
    \lambda\,\gamma\,a.\,\coe_{\bB}\,(\ult,\refl)\,(A\,(\ba\,(\ult^{-1})\,(\ba_{\ext}^{-1}(\gamma,a))))
\end{alignat*}
Likewise, $\coh$-s are given by backwards-forwards $\coh$-s.

$\boldsymbol{\appi : \Tm\,\Gamma\,(\Pii\,a\,B)\ra
  \Tm\,(\Gamma\,\ext\,\El\,a)\,B}$ can be defined as currying of the underlying
functions, and $\boldsymbol{\lami}$ as uncurrying.

\subsection{External Function Space}

For $A : \Set_j$ and $\bB : A \ra \bTy\,\bGamma$, we define $\bPie\,A\,\bB
\boldsymbol{:} \bTy\,\bGamma$ as the $A$-indexed direct product of $\bB$. Since
the indexing is given by a metatheoretic function, every component is given in the
evident pointwise way.

\subsection{Infinitary Function Space}
\label{sec:infinitaryfunction}

For $A : \Set_j$ and $\bb : A \ra \bTm\,\bGamma\,\bU$, we aim to define
$\bPiinf\,A\,\bb \boldsymbol{:} \bTm\,\bGamma\,\bU$. The underlying functions
are:
\begin{alignat*}{3}
  & (\bPiinf\,A\,\bb)\,\ulGamma    &&:= (a : A)\ra \bb\,a\,\ulGamma\\
  & (\bPiinf\,A\,\bb)\,\ulsigma    &&:= \lambda\,a.\, \bb\,a\,\ulsigma\\
  & (\bPiinf\,A\,\bb)\,\ulA\       &&:= \lambda\,\Gamma.\,(a : A)\ra \bb\,a\,\ulA\,(\Gamma\, a)\\
  & (\bPiinf\,A\,\bb)\,\ult        &&:= \lambda\,a.\, \bb\,a\,\ult
\end{alignat*}
The preservation morphisms are as follows. Note that $\emptycon_{\bU} = \top$ and $\ext_{\bU}$ is
metatheoretic $\Sigma$.
\begin{alignat*}{3}
  &(\bPiinf\,A\,\bb)_{\emptycon}^{-1} && : \top\ra (\bPiinf\,A\,\bb)\,\emptycon\\
  &(\bPiinf\,A\,\bb)_{\emptycon}^{-1} && := \lambda\,\_\,a.\,(\bb\,a)_{\emptycon}^{-1}\,\tt\\
  &(\bPiinf\,A\,\bb)_{\ext}^{-1} && : (\Gamma : (\bPiinf\,A\,\bb)\,\ulGamma)\times((\bPiinf\,A\,\bb)\,\ulA\,\Gamma)\\
  & && \hspace{0.5em}\ra (\bPiinf\,A\,\bb)\,(\ulGamma \ext \ulA)\\
  & (\bPiinf\,A\,\bb)_{\ext}^{-1} && := \lambda\,(\Gamma,A)\,a.\,(\bb\,a)_{\ext}^{-1}(\Gamma\,a,\,A\,a)
\end{alignat*}

The preservation of $\emptycon$ and $\blank\ext\blank$ here is in fact the main
point of divergence from \cite{kaposi2019constructing}. In ibid., substitutions
and terms are modeled as strict morphisms and types as displayed CwFs (with no
iso-cleaving). However, it is not the case that $(\bPiinf\,A\,\bb)\,\emptycon =
\top$, which is the statement of strict $\emptycon$-preservation. The left side
reduces to $(a : A)\ra \bb\,a\,\emptycon$, which is isomorphic to $\top$ but not
strictly equal to it. Likewise for $\ext$-preservation.

Hence, we are forced to
interpret terms as weak sections, which in turn forces us to interpret
types as isofibrations, since type substitution requires iso-cleaving.

\subsection{Identity}
For $\bt$ and $\bu$ in $\bTm\,\bGamma\,(\bEl\,\ba)$, we define $\bId\,\bt\,\bu
\boldsymbol{:} \bTm\,\bGamma\,\bU$ as expressing pointwise equality of weak
sections.
\begin{alignat*}{3}
& (\bId\,\bt\,\bu)\,\ulGamma &&:= (\bt\,\ulGamma = \bu\,\ulGamma)\\
& (\bId\,\bt\,\bu)\,\ulA     && := \lambda\,e.\, (\bt\,\ulA = \bu\,\ulA)
\end{alignat*}
Above, $\bt\,\ulA = \bu\,\ulA$ is well-typed because of $e :
\bt\,\ulGamma = \bu\,\ulGamma$. For substitutions, we have to complete a square
of equalities:
\begin{alignat*}{3}
  (\bId\,\bt\,\bu)\,(\ulsigma : \Sub\,\ulGamma\,\ulDelta) : (\bt\,\ulGamma = \bu\,\ulGamma) \ra
       (\bt\,\ulDelta = \bu\,\ulDelta)
\end{alignat*}
This can be given by $\bt\,\ulsigma : \ba\,\ulsigma\,(\bt\,\ulGamma) =
\bt\,\ulDelta$ and $\bu\,\ulsigma : \ba\,\ulsigma\,(\bu\,\ulGamma) =
\bu\,\ulDelta$. The action on terms is analogous. We omit preservation morphisms
here as they are straightforward. Like $\bPiinf$, $\bId$ also does not support
strict preservation of $\emptycon$ and $\ext$. Equality reflection and $\boldsymbol{\refl :}
\bId\,\bt\,\bt$ are also evident.

With this, we have defined the $\bM_{i,j} : \ToS_{\max(i+1,j)+1, j}$ model that
we set out to define in Section \ref{sec:overview}.

\section{Model Theory of the Theory of Signatures}
\label{sec:tossig}

At this point, we only have a notion of algebra for ToS, from Definition
\ref{def:tos}. In the following sections, we would also like to talk about
initial ToS-algebras and ToS-induction. We get these notions by giving a QIIT
signature for ToS, and interpreting it in the $\bM$ model from the previous
section.

\begin{definition}[Signature for ToS] For each level $j$,
we define $\ToSSig_j : \Sig_{j+i}$, as the signature for the theory of
signatures with external sets in $\Set_j$. This is a large and infinitary QIIT
signature, as we have $\Pie$ and $\Piinf$ abstracting over $A : \Set_j$ and
branching with $A \ra \Ty\,\Gamma$ and $A \ra \Tm\,\Gamma\,\U$ respectively. We
present an excerpt from $\ToSSig_j$ below.
\begin{alignat*}{3}
  & \emptycon && \ext (Con : \U)\\
  &           && \ext (Sub : Con\arri Con \arri \U)\\
  &           && \ext (Ty  : Con\arri\U)\\
  &           && \ext (Tm  : \Pii(\Gamma : Con)(\Ty\appitt\Gamma \arri \U))\\
  &           && ...\\
  &           && \ext (\Piinf : \Pii(\Gamma : Con)\\
  &           && \hspace{4.2em}(\Pie \Set_j (\lambda\,A.\,(A \arrinf \Ty\appitt\Gamma)\arri\El(\Ty\appitt\Gamma))))\\
  &           && ...
\end{alignat*}
\end{definition}

Now, for each $i$, the interpretation of $\ToSSig_j$ in $\bM_{i,j+1}$ yields an
flCwF $\bGamma$ such that $\Con_{\bGamma} = \ToS_{i,j}$. In short, we can
recover ToS algebras from the semantics of $\ToSSig$. This follows by
computation of the interpretation and the fact that $\ToSSig$ is precisely the
internal representation of $\ToS$. Hence, we have self-description modulo the
bumping of the $j$ level. Also, as we get an flCwF of $\ToS_{i,j}$-algebras, we
can use Definition \ref{def:induction} for the notion of $\ToS$-induction.

\emph{Remark.} By the definition of $\bemptycon$ and
$\boldsymbol{\blank\ext\blank}$, the types of algebras computed by $\bM$ are
always left-nested iterated $\Sigma$-types which start with $\top$. Hence, we
need to require that Definition \ref{def:tos} is similarly left-nested and
starts with $\top$, in order to make the match strict.

\section{Term Models of QIITs}
\label{sec:termmodels}

In this section we construct QIITs from initial ToS-algebras. For this, we need
to assume the existence of such algebras.

\subsection{Assuming Syntax for the Theory of Signatures}

\begin{lemma}[Cumulativity of $\ToS$]\label{lemma:cumulativity}
If $i \leq i'$, then $\ToS_{i,j} \leq \ToS_{i', j}$.  This follows from the
definition of ToS and the subtyping rules in Section
\ref{sec:cumulativity}. \qed
\end{lemma}

\textbf{Assumption.} For each level $j$ and $k$ such that $j+1 \leq k$, we
assume the existence of $\syn_j : \ToS_{j+1,j}$, and we assume that $\syn_j$,
considered as an element of $\ToS_{k,j}$ by Lemma \ref{lemma:cumulativity}, is
inductive in the sense of Definition \ref{def:induction}.

We explain this assumption. The syntax for the theory of signatures is
postulated at the lowest possible level $\ToS_{j+1,j}$. This is the lowest
because signatures may contain $A : \Set_j$ types, and since we want to view the
syntax as freely generated, its inductive sorts must be large enough to contain
the $A$ types. Otherwise we would run into Russell's paradox. Then, the
induction assumption says that we have induction at all levels larger than
$j+1$.

\begin{example}
We have $\syn_0 : \ToS_{1, 0}$, which is the syntax of closed QIIT
signatures. We want to define a function $\mathsf{length} : \Con_{\syn_0} \ra
\mathbb{N}$ by induction, which returns the length of a syntactic context as a
metatheoretic natural number. To this end, we define a displayed ToS over
$\syn_0$, where $\Con$ is defined as constantly $\mathbb{N}$, every other sort
is defined as constantly $\top$, $\emptycon$ is defined as $0$ and $\Gamma\ext
A$ is defined as $\Gamma + 1$. By the induction assumption, we get a ToS-section
from $\syn_0$ to the displayed model, whose action on contexts is exactly the
$\mathsf{length}$ function. Note that the induction assumption requires that the
displayed model is at least at level 1, but this is not problematic because by
cumulativity $\mathbb{N} : \Set_1$.
\end{example}

For every $M : \ToS_{j+1,j}$, there is a unique strict $\ToS$-morphism from
$\syn_j$ to $M$. This follows from the induction assumption on $\syn_j$ and
Theorem \ref{thm:initialind}. We denote this morphism as
$\llbracket\blank\rrbracket_M$.  For example, given $\Gamma : \Con_\syn$, we
have $\llbracket\Gamma\rrbracket_M : \Con_M$.  Also, for every displayed
$\ToS$-model $M$ over $\syn_j$, there is a strict $\ToS$-section of $M$. We also
denote this as $\llbracket\blank\rrbracket_M$, so e.g.\ for $\ulGamma :
\Con_\syn$ we have $\llbracket\ulGamma\rrbracket_M : \Con_M\,\ulGamma$.

With $\syn$ at hand, we can use an alternative, more conventional representation of
signatures.

\begin{definition}
We define $\SynSig_j : \Set_{j+1}$, the type of \emph{syntactic signatures} at
$j$, as $\Con_{\syn_j}$.
\end{definition}

We can convert a signature to a syntactic one by interpreting it in $\syn_j$,
and we can convert in the other direction by using ToS-induction to interpret a
$\Gamma : \Con_{\syn_j}$ in an arbitrary ToS model. This is merely a logical
equivalence, external to cETT (because of universe polymorphism), and not an
isomorphism.

\subsection{Useful Model Fragments of $\bM$}

In the following, we will need three model fragments of $\bM$, which can be used
to compute notions of algebras, displayed algebras and sections respectively for
each syntactic signature. This is a rephrasing of the $\blank^{\A}$,
$\blank^{\D}$ and $\blank^{\S}$ interpretations in
\cite{kaposi2019constructing}, where they are discussed at more length.

\begin{definition}[The $\Set$ model of $\ToS$] For each $i$ and $j$, we have $\bA : \ToS_{\max(i+1, j)+1, j}$, which can be given by restricting the $\bM_{i,j}$ model
of Section \ref{sec:semantics} so that we only have the first $\Con$ components
in the interpretations for contexts, substitutions, types, terms, and we only
have actions on contexts in the interpretations of term and substitution
formers. Hence, we have:
\begin{alignat*}{3}
  & \Con_{\bA} && = \Set_{\max(i+1, j)} \\
  & \Sub_{\bA}\,\Gamma\,\Delta && = \Gamma \ra \Delta\\
  & \Ty_{\bA}\,\Gamma && = \Gamma \ra \Set_{\max(i+1, j)}\\
  & \Tm_{\bA}\,\Gamma\,A && = (\gamma : \Gamma) \ra A\,\gamma
\end{alignat*}
Now, for some $\Gamma : \Con_{\syn_j}$, the type of $\Gamma$-algebras at level
$i$ is given by $\llb\Gamma\rrb_{\bA}$, where we implicitly lift $\syn_j :
\ToS_{j+1, j}$ to $\ToS_{\max(i+1,j)+1, j}$. E.g.\, $\llb\NatSig\rrb_{\bA}$
yields a left-nested $\Sigma$-type of pointed sets with an endofunction. Also,
$\llb\Gamma\rrb_{\bM}$ extends $\llb\Gamma\rrb_{\bA}$ to an flCwF of
$\Gamma$-algebras, and $\llb\Gamma\rrb_{\bA} = \Con_{\llb\Gamma\rrb_{\bM}}$.
\end{definition}

\begin{definition}[Logical predicate model of $\ToS$ over the $\Set$ model]
For each $i$ and $j$ level we have $\bD$, which is a displayed $\ToS$ model over
$\bA$. This model, analogously to $\bA$, is given by restricting $\bM_{i,j}$ to
the $\Ty$ components everywhere, corresponding to types or actions on
types. Hence, we have:
\begin{alignat*}{3}
  & \Con_{\bD}\,\ulGamma && = \ulGamma \ra \Set_{\max(i+1, j)} \\
  & \Sub_{\bD}\,\Gamma\,\Delta\,\ulsigma && = (\ulgamma : \ulGamma) \ra \Gamma\,\ulgamma \ra \Delta\,(\ulsigma\,\ulgamma)\\
  & \Ty_{\bD}\,\Gamma\,\ulA && = (\ulgamma : \ulGamma) \ra \Gamma\,\ulgamma \ra \ulA\,\ulgamma \ra \Set_{\max(i+1, j)}\\
  & \Tm_{\bD}\,\Gamma\,A\,\ult && = (\ulgamma : \ulGamma)(\gamma : \Gamma\,\ulgamma)
    \ra A\,\ulgamma\,\gamma\,(\ult\,\ulgamma)
\end{alignat*}
\end{definition}

For $\Gamma : \SynSig_j$, the type of displayed $\Gamma$-algebras at level $i$
over some $\ulgamma : \llb\Gamma\rrb_{\bA}$ is given by
$\llb\Gamma\rrb_{\bD}\,\ulgamma$. Here, we also implicitly lift $\Gamma$ to live
in the appropriately sized $\syn$. In other words, $\llb\Gamma\rrb_{\bD}$ yields
the notion of types in the flCwF of $\Gamma$-algebras given by
$\llb\Gamma\rrb_{\bM}$, so we have $\llb\Gamma\rrb_{\bD} = \Ty_{\llb\Gamma\rrb_{\bM}}$.

\begin{definition}[Displayed algebra section model of $\ToS$] Analogously to
$\bA$ and $\bD$, for each $i$ and $j$ levels we define $\bS$ as a displayed
$\ToS$ model over the total model of $\bD$, which is given by restricting
$\bM_{i,j}$ to the $\Tm$ components, corresponding to interpretations of terms and
actions on terms.

For $\ulgamma : \llb\Gamma\rrb_{\bA}$ and $\gamma :
\llb\Gamma\rrb_{\bD}\,\ulgamma$, the type of $\Gamma$-sections at level $i$ is
computed as $\llb\Gamma\rrb_{\bS}\,\ulgamma\,\gamma$, and we have
$\llb\Gamma\rrb_{\bS} = \Tm_{\llb\Gamma\rrb_{\bM}}$.
\end{definition}

\subsection{Term Algebras}

The basic idea is that initial algebras can be built from the terms of $\syn_j$.
For example, consider the syntactic signature for natural numbers:
\[
\NatSig := \emptycon\ext(N : \U)\ext(\mi{zero} : \El\,N)\ext (\mi{suc} : N
\arri \El\,N)
\]
The type $\Tm_{\syn}\,\NatSig\,(\El_{\syn}\,N)$ is isomorphic to the usual type
of natural numbers, since, intuitively, such terms can only be built from
iterated usage of $\mi{zero}$ and $\mi{suc}$. We build a term algebra for each
signature in this manner.

\begin{definition}[Term algebra construction]
For each syntactic signature $\ulOmega : \SynSig_j$, we define a displayed $\ToS$
model over $\syn_j$, named $\bT_{\ulOmega}$. The underlying sets are as follows:
\begin{alignat*}{3}
  & \Con_{\bT_{\ulOmega}}\,\ulGamma &&:= \Sub\,\ulOmega\,\ulGamma \ra \llb\ulGamma\rrb_{\bA}\\
  & \Sub_{\bT_{\ulOmega}}\,\Gamma\,\Delta\,\ulsigma &&:= (\ulnu : \Sub\,\ulOmega\,\ulGamma)
  \ra \Delta\,(\ulsigma\circ\ulnu) \simeq \llb\ulsigma\rrb_{\bA}\,(\Gamma\,\ulnu)\\
  & \Ty_{\bT_{\ulOmega}}\,\Gamma\,\ulA &&:= (\ulnu : \Sub\,\ulOmega\,\ulGamma) \ra
    \Tm\,\ulOmega\,(\ulA[\ulnu]) \ra \llb\ulA\rrb_{\bA}\,(\Gamma\,\ulnu)\\
  & \Tm_{\bT_{\ulOmega}}\,\Gamma\,A\,\ult &&:= (\ulnu : \Sub\,\ulOmega\,\ulGamma)\ra
    \llb\ulA\rrb_{\bA}\,\ulnu\,(\ult[\ulnu]) \simeq_{\id} \llb\ult\rrb_{\bA} (\Gamma\,\ulnu)
\end{alignat*}

Above, the $\simeq$ in the definition of $\Sub_{\bT_{\ulOmega}}$ is a context isomorphism
in $\llb\ulDelta\rrb_{\bM}$, which is the flCwF of $\ulDelta$-algebras. The
$\simeq_{\id}$ in $\Ty_{\bT_{\ulOmega}}$ is a vertical context isomorphism in the displayed
flCwF given by $\llb\ulA\rrb_{\bM}$.

So far, the underlying sets in $\bT_{\ulOmega}$ are similar to what was given in
\cite{kaposi2019constructing} in the construction of term algebras, but there is
an important difference: in ibid.\ strict equalities are used instead of
isomorphisms. In our case, isomorphisms are necessary once again because of
infinitary functions types and our identity type; we shall see this shortly. The
universe is interpreted as follows:
\begin{alignat*}{3}
  & \U_{\bT_{\ulOmega}} : (\ulnu : \Sub\,\ulOmega\,\ulGamma)(\ult : \Tm\,\ulOmega\,\U)\ra
              \Set_{j+1}\\
  & \U_{\bT_{\ulOmega}}\,\ulnu\,\ult := \Tm\,\ulOmega\,(\El\,\ult)\\
  & \El_{\bT_{\ulOmega}}\,a : (\ulnu : \Sub\,\ulOmega\,\ulGamma)
             (\ult : \Tm\,\ulOmega\,(\El\,(\ula[\ulnu])))\ra
              \llb\ula\rrb_{\bA}\,(\Gamma\,\ulnu)\\
  & \El_{\bT_{\ulOmega}}\,a\,\ulnu\,\ult := (a\,\ulnu)\,\ult
\end{alignat*}
Hence, a syntactic $\ult : \Tm\,\ulOmega\,\U$ is interpreted as a set of terms
with type $\El\,\ult$. In the interpretation of $\El$, note that
\[
  a\,\ulnu : \U_{\bT_{\ulOmega}}\,\ulnu\,(\ula[\ulnu]) \simeq_{\id}
                 \llb\ula\rrb_{\bA}\,(\Gamma\,\ulnu)
\]
hence
\[
a\,\ulnu : \Tm\,\ulOmega\,(\El\,(\ula[\ulnu])) \simeq_{\id}
           \llb\ula\rrb_{\bA}\,(\Gamma\,\ulnu)
\]
The $\simeq_{\id}$ above is just an isomorphism of sets, since it lives in
$\llb\U\rrb_{\bM}$ which was given as the flCwF of sets in Section
\ref{sec:universe}. This above isomorphism is a good summary of the
construction: the interpretation of a $\ula : \Tm\,\ulOmega\,\U$ in the term
algebra is isomorphic to a set of terms.

Inductive functions are interpreted by transport along such isomorphism:
\[
  \Pii_{\bT_{\ulOmega}}\,a\,B\,\ulnu\,\ult := \lambda\,\alpha.\,
         B\,(\ulnu,\,(a\,\ulnu)^{-1}\,\alpha)\,(\ult\,\appitt\,((a\,\ulnu)^{-1}\,\alpha))
\]

For the infinitary function space, we need the following, where $\simeq_{\id}$ is
again set isomorphism.
\begin{alignat*}{3}
  & \Piinf_{\bT_{\ulOmega}}\,A\,b\,\ulnu : \Tm\,\ulOmega\,(\El\,(\Piinf\,A\,(\lambda\,\alpha.\,(b\,\alpha)[\ulnu])))\\
  & \hspace{4.2em}\simeq_{\id} ((\alpha : A) \ra \llb\alpha\rrb_{\bA}\,(\Gamma\,\ulnu))
\end{alignat*}
This can be given using the natural isomorphism consisting of $\appinf$ and
$\laminf$. However, the sides are not strictly equal. For the identity type, we
build the following isomorphism using equality reflection.
\begin{alignat*}{3}
  & \Id_{\bT_{\ulOmega}}\,a\,t\,u\,\ulnu : \Tm\,\ulOmega\,(\El\,(\Id\,(\ula[\ulnu])\,(\ult[\ulnu])\,(\ulu[\ulnu]))\\
  & \hspace{4.2em}\simeq_{\id} (\llb\ult\rrb_{\bA}\,(\Gamma\,\ulnu) =   \llb\ulu\rrb_{\bA}\,(\Gamma\,\ulnu))
\end{alignat*}

We omit the rest of the definition of $\bT_{\ulOmega}$. The interpretations of
equations in the CwF and the type formers are fairly technical, and we also
need to utilize iso-cleaving to interpret type substitution and substitution
laws. However, the basic shape of the model remains similar to
\cite{kaposi2019constructing}.
\end{definition}

Now, we can build the term algebra for $\ulOmega$ by taking
$\llb\ulOmega\rrb_{\bT_{\ulOmega}}\,\id$, which has type $\llb\ulOmega\rrb_{\bA}$.

\emph{Remark.} If we start with a syntactic signature at level $j$, then the
underlying sets in the term algebra are all in $\Set_{j+1}$. Hence, the term
algebra for $\NatSig : \SynSig_0$ has an underlying set in $\Set_1$. This is
a bit inconvenient, since normally we would have natural numbers in $\Set_0$. Our
current term model construction cannot avoid this level bump, since $\syn_j$ is
necessarily large, and we do not have a way to construct a small set from
a large set of terms. Perhaps this would be possible with a \emph{resizing rule}
\cite{voevodsky2011resizing}. Also, if we only consider closed finitary QIITs,
with no possibility of referring to external types in signatures, then we can
modify the current term model construction so that we always build sets in
$\Set_0$. This would cover natural numbers and most dependent type theories.

\subsection{Cumulativity of Algebras}

We would like to show that term algebras are initial, but we want to do this
\emph{on all universe levels}, i.e.\ that term algebras are initial when lifted
to any higher level. This requires showing that QII algebras are
cumulative. We do this by induction on syntactic signatures.

\begin{definition}[Cumulativity model]
We assume $j$, $k$ and $l$ levels such that $j+1 \leq k$, $j+1 \leq l$ and $k
\leq l$. We define a displayed model over $\syn_j : \ToS_{j+1, j}$ lifted to
$\ToS_{l, j}$. In the following, we notate the level of algebras computed by
$\llb\blank\rrb_{\bA}$ with an extra index, as in
$\llb\ulGamma\rrb_{\bA_k}$. The underlying sets of the model are as follows.
\begin{alignat*}{3}
  & \Con\,\ulGamma && := \Subtype\,\llb\ulGamma\rrb_{\bA_k}\,\llb\ulGamma\rrb_{\bA_l}\\
  & \Sub\,\Gamma\,\Delta\,\ulsigma && := (\gamma : \llb\ulGamma\rrb_{\bA_k})\ra
  \llb\ulsigma\rrb_{\bA_k}\,\gamma = \llb\ulsigma\rrb_{\bA_l}\,\gamma\\
  & \Ty\,\Gamma\,\ulA &&:= (\gamma : \llb\ulGamma\rrb_{\bA_k})\ra
      \Subtype\,(\llb\ulA\rrb_{\bA_k}\,\gamma)\,(\llb\ulA\rrb_{\bA_l}\,\gamma)\\
  & \Tm\,\Gamma\,A\,\ult &&:= (\gamma : \llb\ulGamma\rrb_{\bA_k})\ra
     \llb\ult\rrb_{\bA_k}\,\gamma = \llb\ult\rrb_{\bA_l}\,\gamma
\end{alignat*}
The rest of the model is straightforward to define. Now, it follows from the
induction assumption for $\syn$ and the reflection rule for $\Subtype$ in
Section \ref{sec:cumulativity}, that $\llb\ulGamma\rrb_{\bA_k} \leq
\llb\ulGamma\rrb_{\bA_l}$.
\end{definition}

\subsection{Term Algebras Support Induction}

\begin{definition}
We assume $j$ and $k$ such that $j + 1 \leq k$, and we also assume $\ulOmega :
\SynSig_j$ and $\gamma :
\llb\ulOmega\rrb_{\bD_k}\,(\llb\ulOmega\rrb_{\bT_{\ulOmega}}\,\id)$.  Hence,
$\gamma$ is a displayed $\ulOmega$-algebra over the term algebra, at level
$k$. We are using the cumulativity of $\ulOmega$ here to lift the term algebra
appropriately.  We aim to show that $\gamma$ has a section. We define a
displayed model over $\syn_j$ lifted to $\ToS_{k, j}$, which we name
$\bI_{\ulOmega}$. The underlying sets are:
\begin{alignat*}{3}
  & \Con_{\bI_{\ulOmega}}\,\ulGamma &&:= (\ulnu : \Sub\,\ulOmega\,\ulGamma)\ra
  \llb\ulGamma\rrb_{\bS}\,(\llb\ulnu\rrb_{\bA}\,(\llb\ulOmega\rrb_{\bT_{\ulOmega}}\,\id))\,\gamma\\
  & \Sub_{\bI_{\ulOmega}}\,\Gamma\,\Delta\,\ulsigma && := (\ulnu : \Sub\,\ulOmega\,\ulGamma)\ra
    \Delta\,(\ulsigma\circ\ulnu) = \llb\ulsigma\rrb_{\bS}\,(\Gamma\,\ulnu)\\
  & \Ty_{\bI_{\ulOmega}}\,\Gamma\,\ulA &&:=
  (\ulnu : \Sub\,\ulOmega\,\ulGamma)(t : \Tm\,\ulOmega\,(\ulA[\ulnu]))\\
  & && \hspace{1.2em}\ra
  \llb\ulA\rrb_{\bS}\,(\llb\ult\rrb_{\bA}\,(\llb\ulOmega\rrb_{\bT_{\ulOmega}}\,\id))\,
  (\llb\ult\rrb_{\bD}\,\gamma)\,(\Gamma\,\ulnu)\\
  & \Tm_{\bI_{\ulOmega}}\,\Gamma\,A\,\ult &&:= A\,\ulnu\,(\ult[\ulnu]) = \llb\ult\rrb_{\bS}\,(\Gamma\,\ulnu)
\end{alignat*}

Here, there is no essential change compared to \cite{kaposi2019constructing},
and we follow ibid.\ in the definition of $\bI_{\ulOmega}$. The reason is that
although we have weakened strict algebra equality to isomorphism, in the current
construction we only have to show equalities of substitutions and terms, which
we do not need to weaken (and they cannot be sensibly weakened anyway).
\end{definition}

\begin{theorem}[Initiality of term algebras]
For each $j$ and $k$ such that $j + 1 \leq k$, and $\ulOmega : \SynSig_j$, the
term algebra given by $\llb\ulOmega\rrb_{\bT_{\ulOmega}}\,\id$ is initial at
level $k$.
\end{theorem}
\begin{proof}
For each $\gamma :
\llb\ulOmega\rrb_{\bD_k}\,(\llb\ulOmega\rrb_{\bT_{\ulOmega}}\,\id)$, we have
$\llb\ulOmega\rrb_{\bI_{\ulOmega}}\,\id :
\llb\ulGamma\rrb_{\bS}\,(\llb\ulOmega\rrb_{\bT_{\ulOmega}}\,\id)\,\gamma$. Hence,
term algebras are inductive in the sense of Definition \ref{def:induction}, and
by Theorem \ref{thm:initialind} they are also initial.
\end{proof}

\section{Related Work}
\label{sec:relatedwork}

Cartmell \cite{gat} defines generalized algebraic theories (GATs) using
type-theoretic syntax. Compared to our QII signatures, he supports infinite
signatures and sort equations but does not cover infinitary constructors or
recursive equations. A way to encode sort equations in our system is using
isomorphisms instead of equalities. In contrast to our algebraic definition,
Cartmell's signatures are given by presyntax, named variables and typing
relations, there is no explicit model theory provided for signatures, and no
explicit term model construction is given. Cartmell focuses instead on showing
that contextual categories serve as classifying categories for GATs.

A more semantic approach to QIITs is given by Altenkirch et
al.\ \cite{qiits}. They generalize the initial algebra semantics of inductive
types to QIITs by considering towers of functors and building complete
categories of algebras from them. Their notion of signature does not enforce
strict positivity, hence describes a larger class of QII signatures. They show
equivalence of initiality and induction, but the lack of a positivity
restriction prevents construction of initial algebras.

The work of Kaposi et al. \cite{kaposi2019constructing} is the direct precursor
of our work. They do not consider infinitary constructors or constructors with
recursive equations, which makes their semantics considerably simpler. They also
do not provide a model theory of signatures, instead they assume signatures as
an ad-hoc QIIT.

Higher inductive types (HITs) are generalizations of QIITs in settings with
proof-relevant identity types. They were introduced before QIITs
\cite{HoTTbook}. \cite{hiit} describes a syntax for higher inductive-inductive
types using a theory of signatures similar to ours, but it does not construct
categories of algebras and initial algebras. Semantics for different subclasses
of HITs are given by
\cite{lumsdaineShulman,10.1145/3209108.3209130,nielsmsc,sojakova,moeneclaey}. Cubical
type theories were shown to support some HITs in a computational way
\cite{cubicalhits,10.1145/3290314}.

Our notion of displayed CwF is an extension of displayed categories
\cite{displayedcats}, although in a setting with UIP.

\section{Conclusions and Further Work}
\label{sec:conclusion}

An important motivation of the current work was to use QIITs as a framework for
algebraic theories, with the metatheory of type theories in mind as a key
application. We would prefer QIITs to
\begin{itemize}
\item Be formally precise.
\item Not gloss over issues of size.
\item Be rich enough to cover most type theories in the wild, including the theory
      of QIIT signatures.
\item Be direct enough, so that signatures for type theories can
      be written out without excessive encoding overhead.
\item Be suitable for practical implementation in proof assistants.
\item Be reducible to a minimal set of basic type formers.
\end{itemize}

With the current work, we have improved the state of QIITs with respect to
the above criteria. However, a number of open research problems remain.

With regards to the expressiveness of QIIT signatures, we do not yet
support sort equations, i.e. equations of elements of $\Tm\,\Gamma\,\U$ in
signatures. Sort equations are included in Cartmell's generalized algebraic
theories \cite{gat}, and they appear to be highly useful for giving an algebraic
representation for Russell-style universes and cumulative universes
\cite{sterling2019algebraic}. We leave this for future work, but we note that
the current isofibration-based semantics does not work in the presence of strict
sort equations, since they are not invariant under isomorphism; instead, sort
equations are compatible with the stricter semantics of
\cite{kaposi2019constructing}.

While we have made an effort to shape the syntax and semantics of QIITs to
be amenable to implementation in proof assistants, much needs to be done before
we can have a practical implementation. For one, we would need to consider QIITs
in a type theory where transports along equality proofs compute, and would need
to work out computing transports for QIITs. Cubical Agda has recently made
strides in implementing HITs \cite{vezzosicubical}, but as of now it does not
support computing transports on indexed inductive types.

With regards to the reduction of QIITs to simple type formers, the reduction of
infinitary QIITs appears to be more challenging than the finitary
case. \cite[Section 9]{lumsdaineShulman} shows that infinitary QIITs are not
constructible from inductive types and simple quotients with relations. In the
finitary case, a generalization of the approach in \cite{induction-is-enough}
seems promising; this amounts to a Streicher-style initial algebra construction
\cite{streicher2012semantics} for the theory of finitary QIIT signatures. In
particular, Brunerie et al. \cite{brunerie} have formalized in Agda this
construction for a comparable type theory, using UIP, function extensionality,
propositional extensionality and simple quotient types.

Another line of possible future work would be to explore a more general
functorial style of semantics for QIITs. So far, we considered set-based
1-categorical semantics, which is what we need when we want to reason
inductively about syntaxes of type theories. However, it would be fruitful to
consider algebras in structured categories other than the category of sets.

\paragraph{Acknowledgments.}The first author was supported by the European Union,
co-financed by the European Social Fund (EFOP-3.6.3-VEKOP-16-2017-00002). The
second author was supported by the National Research, Development and Innovation
Fund of Hungary, financed under the Thematic Excellence Programme funding
scheme, Project no. ED18-1-2019-0030 (Application-specific highly reliable IT
solutions), by the New National Excellence Program of the Ministry for
Innovation and Technology, Project no. \'UNKP-19-4-ELTE-874, and by the Bolyai
Fellowship of the Hungarian Academy of Sciences, Project no. BO/00659/19/3.

\bibliography{references}

\begin{thebibliography}{10}
\providecommand{\url}[1]{#1}
\csname url@samestyle\endcsname
\providecommand{\newblock}{\relax}
\providecommand{\bibinfo}[2]{#2}
\providecommand{\BIBentrySTDinterwordspacing}{\spaceskip=0pt\relax}
\providecommand{\BIBentryALTinterwordstretchfactor}{4}
\providecommand{\BIBentryALTinterwordspacing}{\spaceskip=\fontdimen2\font plus
\BIBentryALTinterwordstretchfactor\fontdimen3\font minus
  \fontdimen4\font\relax}
\providecommand{\BIBforeignlanguage}[2]{{%
\expandafter\ifx\csname l@#1\endcsname\relax
\typeout{** WARNING: IEEEtran.bst: No hyphenation pattern has been}%
\typeout{** loaded for the language `#1'. Using the pattern for}%
\typeout{** the default language instead.}%
\else
\language=\csname l@#1\endcsname
\fi
#2}}
\providecommand{\BIBdecl}{\relax}
\BIBdecl

\bibitem{gat}
\BIBentryALTinterwordspacing
J.~Cartmell, ``Generalised algebraic theories and contextual categories,''
  \emph{Ann. Pure Appl. Log.}, vol.~32, pp. 209--243, 1986. [Online].
  Available: \url{https://doi.org/10.1016/0168-0072(86)90053-9}
\BIBentrySTDinterwordspacing

\bibitem{jacobs1993comprehension}
\BIBentryALTinterwordspacing
B.~Jacobs, ``Comprehension categories and the semantics of type dependency,''
  \emph{Theor. Comput. Sci.}, vol. 107, no.~2, pp. 169--207, 1993. [Online].
  Available: \url{https://doi.org/10.1016/0304-3975(93)90169-T}
\BIBentrySTDinterwordspacing

\bibitem{internaltt}
\BIBentryALTinterwordspacing
P.~Dybjer, ``Internal type theory,'' in \emph{Types for Proofs and Programs,
  International Workshop TYPES'95, Torino, Italy, June 5-8, 1995, Selected
  Papers}, ser. Lecture Notes in Computer Science, S.~Berardi and M.~Coppo,
  Eds., vol. 1158.\hskip 1em plus 0.5em minus 0.4em\relax Springer, 1995, pp.
  120--134. [Online]. Available:
  \url{https://doi.org/10.1007/3-540-61780-9\_66}
\BIBentrySTDinterwordspacing

\bibitem{kaposi2019constructing}
\BIBentryALTinterwordspacing
A.~Kaposi, A.~Kov{\'{a}}cs, and T.~Altenkirch, ``Constructing quotient
  inductive-inductive types,'' \emph{{PACMPL}}, vol.~3, no. {POPL}, pp.
  2:1--2:24, 2019. [Online]. Available: \url{https://doi.org/10.1145/3290315}
\BIBentrySTDinterwordspacing

\bibitem{qiits}
\BIBentryALTinterwordspacing
T.~Altenkirch, P.~Capriotti, G.~Dijkstra, N.~Kraus, and F.~N. Forsberg,
  ``Quotient inductive-inductive types,'' in \emph{Foundations of Software
  Science and Computation Structures - 21st International Conference, {FOSSACS}
  2018, Held as Part of the European Joint Conferences on Theory and Practice
  of Software, {ETAPS} 2018, Thessaloniki, Greece, April 14-20, 2018,
  Proceedings}, ser. Lecture Notes in Computer Science, C.~Baier and U.~D.
  Lago, Eds., vol. 10803.\hskip 1em plus 0.5em minus 0.4em\relax Springer,
  2018, pp. 293--310. [Online]. Available:
  \url{https://doi.org/10.1007/978-3-319-89366-2\_16}
\BIBentrySTDinterwordspacing

\bibitem{dijkstra2017quotient}
\BIBentryALTinterwordspacing
G.~Dijkstra, ``Quotient inductive-inductive definitions,'' Ph.D. dissertation,
  University of Nottingham, {UK}, 2017. [Online]. Available:
  \url{http://ethos.bl.uk/OrderDetails.do?uin=uk.bl.ethos.728471}
\BIBentrySTDinterwordspacing

\bibitem{HoTTbook}
\BIBentryALTinterwordspacing
T.~U.~F. Program, \emph{Homotopy Type Theory: Univalent Foundations of
  Mathematics}.\hskip 1em plus 0.5em minus 0.4em\relax Institute for Advanced
  Study, 2013. [Online]. Available: \url{https://homotopytypetheory.org/book/}
\BIBentrySTDinterwordspacing

\bibitem{lumsdaineShulman}
P.~L. Lumsdaine and M.~Shulman, ``Semantics of higher inductive types,''
  \emph{Mathematical Proceedings of the Cambridge Philosophical Society}, p.
  1–50, 2019.

\bibitem{partiality}
\BIBentryALTinterwordspacing
T.~Altenkirch, N.~A. Danielsson, and N.~Kraus, ``Partiality, revisited,'' in
  \emph{Proceedings of the 20th International Conference on Foundations of
  Software Science and Computation Structures - Volume 10203}.\hskip 1em plus
  0.5em minus 0.4em\relax Berlin, Heidelberg: Springer-Verlag, 2017, p.
  534–549. [Online]. Available:
  \url{https://doi.org/10.1007/978-3-662-54458-7_31}
\BIBentrySTDinterwordspacing

\bibitem{cchm}
\BIBentryALTinterwordspacing
C.~Cohen, T.~Coquand, S.~Huber, and A.~M{\"o}rtberg, ``{Cubical Type Theory: A
  Constructive Interpretation of the Univalence Axiom},'' in \emph{21st
  International Conference on Types for Proofs and Programs (TYPES 2015)}, ser.
  Leibniz International Proceedings in Informatics (LIPIcs), T.~Uustalu, Ed.,
  vol.~69.\hskip 1em plus 0.5em minus 0.4em\relax Dagstuhl, Germany: Schloss
  Dagstuhl--Leibniz-Zentrum fuer Informatik, 2018, pp. 5:1--5:34. [Online].
  Available: \url{http://drops.dagstuhl.de/opus/volltexte/2018/8475}
\BIBentrySTDinterwordspacing

\bibitem{angiuli2016computational}
\BIBentryALTinterwordspacing
C.~Angiuli, R.~Harper, and T.~Wilson, ``Computational higher type theory {I:}
  abstract cubical realizability,'' \emph{CoRR}, vol. abs/1604.08873, 2016.
  [Online]. Available: \url{http://arxiv.org/abs/1604.08873}
\BIBentrySTDinterwordspacing

\bibitem{angiuli2018cartesian}
\BIBentryALTinterwordspacing
C.~Angiuli, K.-B. {Hou (Favonia)}, and R.~Harper, ``Cartesian cubical
  computational type theory: Constructive reasoning with paths and
  equalities,'' in \emph{27th {EACSL} Annual Conference on Computer Science
  Logic, {CSL} 2018, September 4-7, 2018, Birmingham, {UK}}, ser. LIPIcs, D.~R.
  Ghica and A.~Jung, Eds., vol. 119.\hskip 1em plus 0.5em minus 0.4em\relax
  Schloss Dagstuhl - Leibniz-Zentrum f{\"{u}}r Informatik, 2018, pp. 6:1--6:17.
  [Online]. Available: \url{https://doi.org/10.4230/LIPIcs.CSL.2018.6}
\BIBentrySTDinterwordspacing

\bibitem{clairambault2014biequivalence}
\BIBentryALTinterwordspacing
P.~Clairambault and P.~Dybjer, ``The biequivalence of locally cartesian closed
  categories and martin-l{\"{o}}f type theories,'' \emph{Mathematical
  Structures in Computer Science}, vol.~24, no.~6, 2014. [Online]. Available:
  \url{https://doi.org/10.1017/S0960129513000881}
\BIBentrySTDinterwordspacing

\bibitem{timany2018cumulative}
\BIBentryALTinterwordspacing
A.~Timany and M.~Sozeau, ``Cumulative inductive types in coq,'' in \emph{3rd
  International Conference on Formal Structures for Computation and Deduction,
  {FSCD} 2018, July 9-12, 2018, Oxford, {UK}}, ser. LIPIcs, H.~Kirchner, Ed.,
  vol. 108.\hskip 1em plus 0.5em minus 0.4em\relax Schloss Dagstuhl -
  Leibniz-Zentrum f{\"{u}}r Informatik, 2018, pp. 29:1--29:16. [Online].
  Available: \url{https://doi.org/10.4230/LIPIcs.FSCD.2018.29}
\BIBentrySTDinterwordspacing

\bibitem{sterling2019algebraic}
J.~Sterling, ``Algebraic type theory and universe hierarchies,'' \emph{arXiv
  preprint arXiv:1902.08848}, 2019.

\bibitem{carette2007finally}
\BIBentryALTinterwordspacing
J.~Carette, O.~Kiselyov, and C.~Shan, ``Finally tagless, partially evaluated:
  Tagless staged interpreters for simpler typed languages,'' \emph{J. Funct.
  Program.}, vol.~19, no.~5, pp. 509--543, 2009. [Online]. Available:
  \url{https://doi.org/10.1017/S0956796809007205}
\BIBentrySTDinterwordspacing

\bibitem{cubicalmodel}
\BIBentryALTinterwordspacing
M.~Bezem, T.~Coquand, and S.~Huber, ``A model of type theory in cubical sets,''
  in \emph{19th International Conference on Types for Proofs and Programs,
  {TYPES} 2013, April 22-26, 2013, Toulouse, France}, ser. LIPIcs, R.~Matthes
  and A.~Schubert, Eds., vol.~26.\hskip 1em plus 0.5em minus 0.4em\relax
  Schloss Dagstuhl - Leibniz-Zentrum f{\"{u}}r Informatik, 2013, pp. 107--128.
  [Online]. Available: \url{https://doi.org/10.4230/LIPIcs.TYPES.2013.107}
\BIBentrySTDinterwordspacing

\bibitem{hiit}
\BIBentryALTinterwordspacing
A.~Kaposi and A.~Kov{\'{a}}cs, ``Signatures and induction principles for higher
  inductive-inductive types,'' \emph{CoRR}, vol. abs/1902.00297, 2019.
  [Online]. Available: \url{http://arxiv.org/abs/1902.00297}
\BIBentrySTDinterwordspacing

\bibitem{birkedal2020modal}
\BIBentryALTinterwordspacing
L.~Birkedal, R.~Clouston, B.~Mannaa, R.~E. M{\o}gelberg, A.~M. Pitts, and
  B.~Spitters, ``Modal dependent type theory and dependent right adjoints,''
  \emph{Mathematical Structures in Computer Science}, vol.~30, no.~2, pp.
  118--138, 2020. [Online]. Available:
  \url{https://doi.org/10.1017/S0960129519000197}
\BIBentrySTDinterwordspacing

\bibitem{displayedcats}
\BIBentryALTinterwordspacing
B.~Ahrens and P.~L. Lumsdaine, ``Displayed categories,'' \emph{Logical Methods
  in Computer Science}, vol.~15, no.~1, 2019. [Online]. Available:
  \url{https://doi.org/10.23638/LMCS-15(1:20)2019}
\BIBentrySTDinterwordspacing

\bibitem{bernardy12parametricity}
\BIBentryALTinterwordspacing
J.~Bernardy, P.~Jansson, and R.~Paterson, ``Proofs for free - parametricity for
  dependent types,'' \emph{J. Funct. Program.}, vol.~22, no.~2, pp. 107--152,
  2012. [Online]. Available: \url{https://doi.org/10.1017/S0956796812000056}
\BIBentrySTDinterwordspacing

\bibitem{johnstone2002sketches}
P.~T. Johnstone, \emph{Sketches of an elephant: A topos theory
  compendium}.\hskip 1em plus 0.5em minus 0.4em\relax Oxford University Press,
  2002, vol.~1.

\bibitem{voevodsky2011resizing}
\BIBentryALTinterwordspacing
V.~Voevodsky, ``Resizing rules, slides from a talk at types2011,'' \emph{At
  author’s webpage}, 2011. [Online]. Available:
  \url{https://www.math.ias.edu/vladimir/sites/math.ias.edu.vladimir/files/2011_Bergen.pdf}
\BIBentrySTDinterwordspacing

\bibitem{10.1145/3209108.3209130}
\BIBentryALTinterwordspacing
S.~Awodey, J.~Frey, and S.~Speight, ``Impredicative encodings of (higher)
  inductive types,'' in \emph{Proceedings of the 33rd Annual ACM/IEEE Symposium
  on Logic in Computer Science}, ser. LICS ’18.\hskip 1em plus 0.5em minus
  0.4em\relax New York, NY, USA: Association for Computing Machinery, 2018, p.
  76–85. [Online]. Available: \url{https://doi.org/10.1145/3209108.3209130}
\BIBentrySTDinterwordspacing

\bibitem{nielsmsc}
N.~van~der Weide, ``Higher inductive types,'' Master's thesis, Radboud
  University, Nijmegen, 2016.

\bibitem{sojakova}
K.~Sojakova, ``Higher inductive types as homotopy-initial algebras,'' in
  \emph{Proceedings of the 42Nd Annual ACM SIGPLAN-SIGACT Symposium on
  Principles of Programming Languages}, ser. POPL '15.\hskip 1em plus 0.5em
  minus 0.4em\relax New York, NY, USA: ACM, 2015, pp. 31--42.

\bibitem{moeneclaey}
\BIBentryALTinterwordspacing
P.~Dybjer and H.~Moeneclaey, ``Finitary higher inductive types in the groupoid
  model,'' in \emph{Proceedings of the Thirty-Fourth Conference on the
  Mathematical Foundations of Programming Semantics, {MFPS} 2018, Dalhousie
  University, Halifax, Canada, June 6-9, 2018}, ser. Electronic Notes in
  Theoretical Computer Science, S.~Staton, Ed., vol. 341.\hskip 1em plus 0.5em
  minus 0.4em\relax Elsevier, 2018, pp. 119--134. [Online]. Available:
  \url{https://doi.org/10.1016/j.entcs.2018.03.019}
\BIBentrySTDinterwordspacing

\bibitem{cubicalhits}
\BIBentryALTinterwordspacing
T.~Coquand, S.~Huber, and A.~M{\"{o}}rtberg, ``On higher inductive types in
  cubical type theory,'' in \emph{Proceedings of the 33rd Annual {ACM/IEEE}
  Symposium on Logic in Computer Science, {LICS} 2018, Oxford, UK, July 09-12,
  2018}, A.~Dawar and E.~Gr{\"{a}}del, Eds.\hskip 1em plus 0.5em minus
  0.4em\relax {ACM}, 2018, pp. 255--264. [Online]. Available:
  \url{https://doi.org/10.1145/3209108.3209197}
\BIBentrySTDinterwordspacing

\bibitem{10.1145/3290314}
\BIBentryALTinterwordspacing
E.~Cavallo and R.~Harper, ``Higher inductive types in cubical computational
  type theory,'' \emph{Proc. ACM Program. Lang.}, vol.~3, no. POPL, Jan. 2019.
  [Online]. Available: \url{https://doi.org/10.1145/3290314}
\BIBentrySTDinterwordspacing

\bibitem{vezzosicubical}
\BIBentryALTinterwordspacing
A.~Vezzosi, A.~M{\"{o}}rtberg, and A.~Abel, ``Cubical agda: a dependently typed
  programming language with univalence and higher inductive types,''
  \emph{{PACMPL}}, vol.~3, no. {ICFP}, pp. 87:1--87:29, 2019. [Online].
  Available: \url{https://doi.org/10.1145/3341691}
\BIBentrySTDinterwordspacing

\bibitem{induction-is-enough}
\BIBentryALTinterwordspacing
A.~Kaposi, A.~Kovács, and L.~Ambroise, ``For finitary induction-induction,
  induction is enough,'' \emph{Submitted to TYPES 2019 post-proceedings}, 2019.
  [Online]. Available:
  \url{https://github.com/amblafont/UniversalII/blob/cwf-syntax/paper/paper.pdf}
\BIBentrySTDinterwordspacing

\bibitem{streicher2012semantics}
T.~Streicher, \emph{Semantics of type theory: correctness, completeness and
  independence results}.\hskip 1em plus 0.5em minus 0.4em\relax Springer
  Science \& Business Media, 2012.

\bibitem{brunerie}
\BIBentryALTinterwordspacing
G.~Brunerie, ``A formalization of the initiality conjecture in agda,'' August
  2019, slides of a talk at the Homotopy Type Theory 2019 Conference, Carnegie
  Mellon University, Pittsburgh, Pennsylvania. [Online]. Available:
  \url{https://guillaumebrunerie.github.io/pdf/initiality.pdf}
\BIBentrySTDinterwordspacing

\end{thebibliography}

\end{document}